\documentclass[sigconf]{acmart}

\usepackage{kotex}
\usepackage{framed}
\usepackage{soul}
\usepackage{lipsum}
\usepackage{amsthm} 
\usepackage{comment}
\usepackage{multirow}
\usepackage{enumitem} 
\usepackage{subcaption}
\usepackage{footnote}
\usepackage{tablefootnote}
\usepackage{caption}
\usepackage{tikz}
\usepackage{stfloats} 
\usepackage{thmtools}
\usepackage[normalem]{ulem}
\useunder{\uline}{\ul}{}
\usepackage[para,online,flushleft]{threeparttable}

\usepackage{titlesec}

\usepackage{pifont}
\newcommand{\cmark}{\ding{51}}%
\newcommand{\xmark}{\ding{55}}%

\usepackage[linesnumbered,boxruled,ruled,noend]{algorithm2e}

\AtBeginDocument{%
  \providecommand\BibTeX{{%
    \normalfont B\kern-0.5em{\scshape i\kern-0.25em b}\kern-0.8em\TeX}}}

\newtheorem{theorem}{Theorem}
\newtheorem{definition}{Definition}
\newtheorem{remark}{Remark}

\newtheorem{example}{Example}
\newtheorem{pDefinition}{Problem Definition}

\newcommand{\NCADR}{\textsf{NCA-DR}}
\newcommand{\FPADMG}{\textsf{FPA-DMG}}

\newcommand{\DMCS}{\textsf{DMCS}}
\newcommand{\NCA}{\textsf{NCA}}
\newcommand{\FPA}{\textsf{FPA}}
\newcommand{\kc}{\textsf{kc}}
\newcommand{\wu}{\textsf{wu2015}}
\newcommand{\icwi}{\textsf{icwi2008}}
\newcommand{\clique}{\textsf{clique}}
\newcommand{\GN}{\textsf{GN}}
\newcommand{\kt}{\textsf{kt}}
\newcommand{\kecc}{\textsf{kecc}}
\newcommand{\huang}{\textsf{huang2015}}
\newcommand{\CNM}{\textsf{CNM}}
\newcommand{\hightruss}{\textsf{hightruss}}
\newcommand{\highcore}{\textsf{highcore}}

\SetCommentSty{mycommfont}
\newcommand{\spara}[1]{\smallskip\noindent{\bf #1}}
\SetKwRepeat{Do}{do}{while}

\DeclareMathOperator*{\argmax}{arg\,max} 
\def\cmark{\tikz\fill[scale=0.4](0,.35) -- (.25,0) -- (1,.7) -- (.25,.15) -- cycle;}

\vskip\baselineskip
\renewcommand\footnotetextcopyrightpermission[1]{} 

\setcopyright{none}
\begin{document}
\fancyhead{}

\title{DMCS : Density Modularity based Community Search}

\author{Junghoon Kim$\S$, Siqiang Luo$\S$, Gao Cong$\S$, Wenyuan Yu$\dag$}
\affiliation{%
  \institution{$\S$School of Computer Science and Engineering, Nanyang Technological University, Singapore}
  \country{$\dag$ Alibaba Group, China}}
\email{{junghoon001@e., siqiang.luo, gaocong@}ntu.edu.sg, wenyuan.ywy@alibaba-inc.com}


\renewcommand{\shortauthors}{Kim, et al.}

\begin{abstract}

Community Search, or finding a connected subgraph (known as a community) containing the given query nodes in a social network, is a fundamental problem. Most of the existing community search models only focus on the internal cohesiveness of a community. However, a high-quality community often has high modularity, which means dense connections inside communities and sparse connections to the nodes outside the community. In this paper, we conduct a pioneer study on searching a community with high modularity. We point out that while modularity has been popularly used in community detection (without query nodes), it has not been adopted for community search, surprisingly, and its application in community search (related to query nodes) brings in new challenges. We address these challenges by designing a new graph modularity function named {\it Density Modularity}. To the best of our knowledge, this is the first work on the community search problem using graph modularity. The community search based on the density modularity, termed as DMCS, is to find a community in a social network that contains all the query nodes and has high density-modularity. We prove that the DMCS problem is NP-hard. To efficiently address DMCS, we present new algorithms that run in log-linear time to the graph size. We conduct extensive experimental studies in real-world and synthetic networks, which offer insights into the efficiency and effectiveness of our algorithms. In particular, our algorithm achieves up to 8.5 times higher accuracy in terms of NMI than baseline algorithms.
\end{abstract}


\maketitle

\section{Introduction}
\vspace{-0.1cm}
Given a graph $G$ and a set of query nodes $Q$, the community search problem aims to find a connected subgraph that contains all the query nodes in $Q$ and satisfies some cohesiveness constraints~\cite{fang2020survey,cui2014local,barbieri2015efficient,huang2014querying}. Recently, this community search problem has attracted extensive research interest in diverse fields with applications such as marketing~\cite{kim2020densely}, recommendation~\cite{fang2017effective}, and social event organization~\cite{sozio2010community}. Existing studies~\cite{cui2014local, barbieri2015efficient, huang2014querying, wu2015robust, sozio2010community, kim2020densely} explore different ways of defining a {\it community}, aiming for efficiently extracting an effective community.

\begin{figure}[t]
\centering
\includegraphics[width=0.99\linewidth]{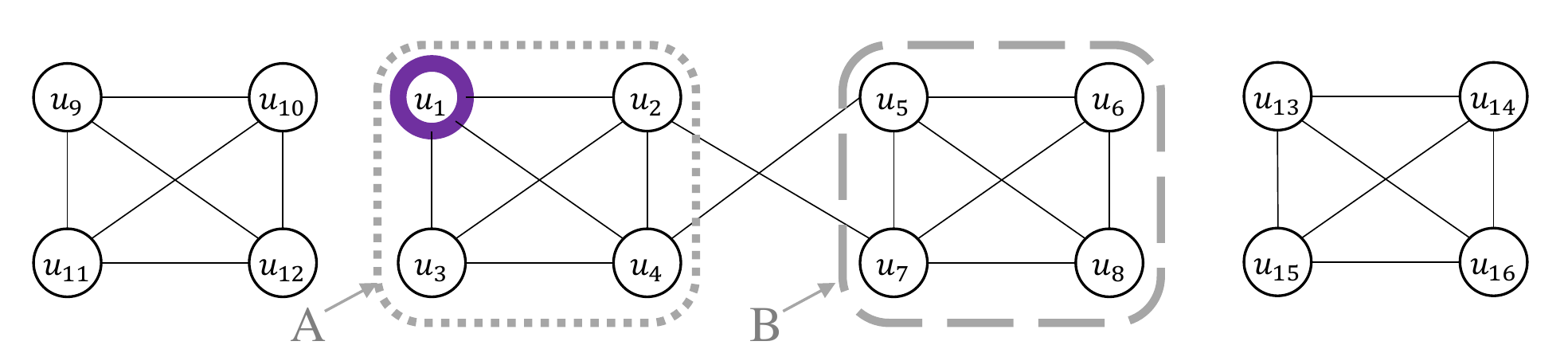}
\vspace{-0.3cm}
\caption{Toy network with communities}
\label{DMCS_fig:example}
\end{figure}

\spara{Motivation.} Many existing community search models are based on the minimum degree~\cite{sozio2010community,cui2014local,barbieri2015efficient, kim2020densely,wang2020efficient, fang2017effective} (also known as $k$-core) and triangle counting~\cite{huang2014querying, jiang2021efficient,zheng2017finding,liu2021efficient,akbas2017truss} (also known as $k$-truss). 
Given a graph and positive integer $k$, a connected subgraph is a $k$-core community if its minimum degree is at least $k$. Similarly, a maximal subgraph $H$ is a $k$-truss community if each edge in $H$ participates in at least $(k-2)$ triangles in $H$. 
Note that both models depend on the input parameter $k$ and the community quality is sensitive to the inherent parameters, although such models are easy to interpret and compute. When $k$ is not properly set, they may fail to return high-quality communities. For example, if every node has at least $3$ neighbor nodes, searching for a $3$-core will return the whole graph. Similarly, if all the edges are involved in at least $2$ triangles, then the whole graph is a $4$-truss. In other words, setting a small $k$ tends to give us very large communities, which are not useful in practice. However, if we set a large $k$, we are at risk of not getting any results because the constraints are stringent. Therefore, finding proper parameters is challenging in these models~\cite{chu2020finding}. 
To alleviate this issue, some approaches return the highest-order core/truss, i.e., the $k$-core or $k$-truss that contains the query nodes while $k$ is maximized. However, such methods may still return undesired results~\cite{sozio2010community} because most of the nodes in real-world networks have low order of core/truss~\cite{shin2018patterns}. As a result, the maximum $k$ is still naturally small.

To address the above problem in community search, we turn to the graph modularity, which is a parameter-free measure. 
Graph modularity~\cite{newman2006modularity} is the fraction of the edges that fall within the given groups minus the expected fraction if edges were distributed at random.
Surprisingly, graph modularity has not been used for community search, although it has been popularly used in community detection~\cite{lim2016blackhole,yang2015defining}, which identifies communities to maximize the graph modularity without the constraints of query nodes.
Unfortunately, using such classic modularity for community search still has several limitations. First,  
the classic modularity suffers from the free-rider effect~\cite{wu2015robust} -- the resultant community may contain many nodes not related to the query nodes;
this effect is illustrated by Figure~\ref{DMCS_fig:example}. Suppose the query node is $u_1$; then a desirable community should be $A$ because it is densely connected internally and sparsely connected externally. However, based on the definition of the classic graph modularity, the community $A\cup B$ has a higher modularity. 
As a result, the whole subgraph $B$ becomes a ``free-rider'' when searching for a community that contains $u_1$. This example also simultaneously implies  the {\it resolution limit problem}~\cite{fortunato2007resolution} which concerns that the community search fails to identify a small-sized community and thus being not able to highlight some important structures. In this example, particularly, subgraph $A$ is an important structure that needs to be discovered.

\spara{New Modularity Definition.}
To mitigate the above problems, we propose the {\it density modularity}, a new modularity for community search, by seamlessly integrating the benefits of two parameter-free classic measures, namely, graph modularity~\cite{newman2006modularity} and graph density~\cite{khuller2009finding}. 
Graph density~\cite{khuller2009finding} is the ratio of the number of edges and the number of nodes. A subgraph with high modularity means that it has dense connections within the community {\it in comparison to} its connections to nodes outside the community. In other words, modularity is more about the {\it relative cohesiveness} comparing the community internals and externals. Subgraph density, in contrast, describes the ratio between the number of subgraph edges and the number of nodes within the subgraph~\cite{khuller2009finding,charikar2000greedy,goldberg1984finding}. In other words, density is more about the {\it absolute cohesiveness} of the community. 
By incorporating the graph density measure into modularity, we are able to capture both the \textit{relative} and \textit{absolute} cohesiveness of a community. 
Specifically, the classic modularity summarizes the difference between the number of internal edges and the expected fraction of random edges, and then the score is divided  by the number of total edges for the normalization.
In the proposed density modularity, we replace the normalization term with the size of resultant community, which reflects the density. 

Notably, the new density modularity is characterized with two salient features. (1)~Density modularity still substantially inherits the benefit of classic modularity because it remains the second multiplicative factor that depicts the contrast of community internals and externals. (2)~This density modularity can be reinterpreted in terms of a classic graph density~\cite{khuller2009finding}, i.e., it can be rewritten as  the difference between the graph density minus the community-sized normalized  fraction of the randomized edges. 
Hence, a high density modularity also typically means a high density. (3)~The mediation of modularity and density alleviates the aforementioned issues. For example, $A\cup B$ in Figure~\ref{DMCS_fig:example} has a smaller density than $A$, and thus, the density modularity based community search will return us $A$ instead of $A\cup B$.

\spara{New Problem.}
In this paper, we define the \underline{D}ensity \underline{M}odularity \underline{C}ommunity \underline{S}earch  ({\DMCS}) problem, which aims to find a community containing all the query nodes such that the density modularity of the identified community is maximized.

In a nutshell, the benefits of using density modularity for {\DMCS} are summarized as follows.

\begin{itemize}[leftmargin=*]
\item Compared with $k$-core and $k$-truss, {\DMCS} is parameter-free, and therefore it does not suffer from the parameter-sensitivity issue; it considers the internal and external edges simultaneously to identify a community. 

\item Compared with the classic graph modularity, it alleviates the free-rider effect and the resolution limit problem. 
A direct benefit is that it avoids finding a giant community. We prove both theoretically and empirically that density modularity has strictly better performance than the classic modularity in alleviating the free-rider effect and the resolution limit problem 
(See Section~\ref{DMCS_sec:dm_property}).
\end{itemize}

Similar to the classic modularity maximization problem~\cite{dasgupta2013complexity,brandes2006maximizing}, the {\DMCS} problem is NP-Hard (Section~\ref{DMCS_sec:prob_def}). Therefore, it is prohibitively expensive to compute an exact solution. On the technical side, there are two main challenges for designing algorithms: (1) the community should be connected; (2) maximizing the density modularity. 
To address the challenges, our general idea is to first find nodes such that removing them will not disconnect the remaining graph.
Among such nodes, we then strategically select the node to be removed with the goal of maximizing the density modularity. We will discuss detailed techniques in Section~\ref{DMCS_sec:alg} on how to optimize the computational complexity under this algorithmic framework.

\spara{Contribution.} Our contributions  are summarized as follows: 
\begin{itemize}[leftmargin=*]
    \item \textit{Problem definition : }  To the best of our knowledge, this is the first work to incorporate both classic graph modularity~\cite{newman2004finding} and graph density~\cite{khuller2009finding} for the community search problem. We propose a new modularity named \textit{density modularity} to capture both the absolute cohesiveness and the relative cohesiveness.  
    \item \textit{Theoretical analysis : } One known issue of the classic modularity 
    is that it suffers from free-rider effects and has the resolution limit problem. In this paper, we rigorously prove that our density modularity significantly alleviates the two problems.
    \item \textit{Designing new algorithms : } Since the {\DMCS} problem is NP-hard, we propose two polynomial  algorithms.
    We show that one of our algorithms runs in only log-linear time to the graph size. 
    \item \textit{Extensive experiments} : By using real-world graphs and synthetic networks, we conduct an extensive experimental study to show the effectiveness and efficiency of our algorithms. 
     
\end{itemize}

\section{Related work}\label{DMCS_sec:relatedwork}
\vspace{-0.1cm}
\subsection{Community Search}\label{DMCS_sec:related_comm}

The community search problem was first proposed by ~\cite{sozio2010community}.
Formally, given a set $S$ of query nodes in $G$, community search aims to discover a connected and dense subgraph, known as a community, that contains all the query nodes in $S$. The community is expected to be cohesive, implying that the nodes within the community are intensively linked between each other. A plethora of community search algorithms are proposed~\cite{cui2014local,yuan2017index, wu2015robust,barbieri2015efficient,fang2017effective, kim2020densely,fang2016effective,huang2014querying,huang2015approximate}, with various definitions of cohesiveness. Below we describe representative ones and more can be found in the survey paper~\cite{fang2020survey}.

Sozio et al.~\cite{sozio2010community} propose to extract a $k$-core community with the minimum degree $k$ and prove the hardness of the problem. They propose a global-search based (GS) algorithm by incorporating a peeling strategy~\cite{charikar2000greedy}. Cui et al.~\cite{cui2014local}  improve the search efficiency with a local search (LS) algorithm, which incrementally expands the community from the query node. 
Observing that the community returned by the GS algorithm~\cite{sozio2010community} can be too large, Barbieri et al.~\cite{barbieri2015efficient} aim to minimize the community size while allowing multiple query nodes. They propose a local greedy search algorithm by incorporating the Steiner-tree algorithm. 
Another category of community model is the triangle-based community model (also known as the $k$-truss model)~\cite{huang2014querying,akbas2017truss,jiang2021efficient,liu2021efficient}. Given an integer $k>2$, the model requires that every edge in the community should be involved in at least $(k-2)$ triangles. 
\textcolor{black}{
Instead of specifying $k$ for $k$-core and $k$-truss community search models, several works aim to find the highest-order core~\cite{sozio2010community, cui2014local} or truss~\cite{huang2015approximate} to find cohesive subgraphs.
Recently, Yao et al.~\cite{yao2021efficient} propose a size-bounded community model by maximizing the minimum degree.  }

\textcolor{black}{
Wu et al.~\cite{wu2015robust} study the existing goodness functions for community search problem. They find that most goodness functions like minimum degree suffer from the free-rider effect problem. 
Compared with  \cite{wu2015robust}, our problem has three main advantages: (1) we do not require any user parameters such as decay factor $c$, a parameter $K$ for controlling search space, or a parameter $\eta$ for controlling the degree of the non-articulation nodes to be removed; (2) we consider both internal and external edges by considering the global structure; (3) the quality of the result does not depend on the location of the query nodes. Note that \cite{wu2015robust} may find low-quality result if a query node is not in the center of a community since it prefers the nodes that are close to the query node.
In addition, the result of \cite{wu2015robust} might be sensitive to the user parameters.
}

In addition, there are studies on attributed community search, where each node is associated with attributes~\cite{huang2017attribute, huang2014querying, fang2016effective, fang2017effective, kim2020densely}. They are orthogonal to our work.

\subsection{Graph Modularity}\label{DMCS_sec:related_modu}

\textcolor{black}{
The graph modularity\cite{newman2004finding} is designed to measure the quality of community detection algorithms. Maximizing the graph modularity is NP-hard~\cite{brandes2006maximizing}, and \cite{dinh2015network} proves the inapproximability for modularity clustering with any (multiplicative) factor $c>0$.
We next discuss several representative modularity optimization algorithms. 
}

The Divisive algorithm~\cite{newman2004finding} is a top-down approach. It iteratively removes important edges to 
find a connected component with the largest graph modularity. 
The Agglomerative algorithm~\cite{clauset2004finding} is a bottom-up approach. It starts from all the singleton communities, and iteratively joins communities in pairs to maximize the modularity.
The Louvain algorithm~\cite{blondel2008fast} is a hierarchical community detection algorithm. All the nodes belonging to the same community are merged into a single giant node and modularity clustering on the condensed graphs is applied. Both community aggregation and modularity clustering are executed until maximum modularity is reached. The Louvain algorithm is one of the best modularity optimization algorithms~\cite{lancichinetti2009community}.

These approaches cannot be directly applied for community search. The reasons are two-fold. First, as they are designed for community detection, they need to compute all the communities spanning the whole network, which is costly. Second,  
it is known that modularity optimization suffers from the resolution limit problem~\cite{fortunato2007resolution}.
To mitigate the resolution limit problem, some approaches~\cite{chen2013measuring, chen2015new, guo2020resolution} have been proposed and they try to extend classic modularity by adding additional terms such as split penalty to avoid finding communities of large size for the community detection problem. 
However, these measures are proposed  for the community detection problem and cannot be used for our problem.

\section{Preliminaries}\label{DMCS_sec:preliminaries}
\vspace{-0.1cm}
A social network is modeled as a graph $G=(V, E)$ with node set $V$ and edge set $E$. 
Following previous studies~\cite{sozio2010community,cui2014local,kim2020densely}, we consider that $G$ is undirected. Given a set of nodes $C\subseteq V$, we denote $G[C]$ as the induced subgraph of $G$ which takes $C$ as its node set and $E[C]=\{(u,v) \in E|u,v \in C\}$ as its edge set. Given a set of nodes $C$, the modularity $CM(G,C)$ is defined as follows.

\begin{definition}
(\underline{Classic modularity of a community}~\cite{newman2004finding,brandes2006maximizing}). 
Given a graph $G=(V,E)$ and a set of nodes $C\subseteq V$, the classic modularity of $C$ is defined as follows. 
\begin{align}
CM(G, C) = \frac{1}{2|E|}(2l_C - \frac{d_C^2}{2|E|})
\end{align} 
where $d_C$ is the sum of degrees of the nodes in $C$, and $l_C$ is the number of internal edges in $G[C]$.
\end{definition}

Given a graph $G$, the classic modularity maximization problem aims to find disjoint communities $\mathcal{C}=\{C_1, C_2, \cdots C_g\}$ such that $\sum_{C\in \mathcal{C}} CM(G, C)$ is maximized. 

\begin{example}\label{DMCS_example:mod_CM_example}

Let us use Figure~\ref{DMCS_fig:example} to illustrate how to compute the classic modularity. Suppose that $u_1$ (purple node) is a query node. We consider two communities $A$ and $A\cup B$ in this example. Given $|E|=26$, $l_{A\cup B} = 14$,  $d_{A\cup B} = 28$, $l_{A} = 6$, and $d_{A}=14$, the  modularity of community $A$ and community $A \cup B$ is computed as follows.
\begin{itemize}[leftmargin=*]
    \item $CM(A)=\frac{1}{52}(12 - \frac{14^2}{52})=  0.158284$
    \item $CM(A\cup B)=\frac{1}{52}(28 - \frac{28^2}{52})= 0.2485207$
\end{itemize}
\end{example}

\noindent
{\bf{Limitations of Classic Modularity for Community Search}}.
To apply the classic modularity definition for community search, we need to address the resolution limit problem~\cite{fortunato2007resolution}. The resolution limit issue of modularity maximization~\cite{fortunato2007resolution} indicates that optimizing the modularity may fail to discover relatively small communities even if small communities are densely connected.
Specifically, the classic modularity compares the number of edges in a community with the expected number of edges. To measure the expected number of edges,  random null model~\cite{chen2014community,muff2005local} is used. The random null model assumes that connections between all pairs of the nodes in a network are uniformly probable~\cite{chen2014community,muff2005local}. However, this assumption 
does not hold when the graph becomes very large. We notice that the expected number edges of two communities becomes smaller when the graph is large. It implies that a small number of edges in two communities may lead to merging the two communities.

Moreover, we observe that using the classic modularity for community search suffers from the free-rider effect~\cite{wu2015robust},  which indicates that the resultant community may contain many nodes irrelevant to query nodes. 
Intuitively, if a community goodness function allows 
irrelevant subgraphs in the resultant community, we refer to the nodes in the irrelevant subgraphs as free riders. For instance, suppose that we use classic graph density ($\frac{|E|}{|V|}$) as the community goodness function and find a solution $C$. Then, if we merge the densest subgraph with the current solution $C$, the graph density will increase. It indicates that the classic graph density suffers from the free-rider effect. 
In Section~\ref{DMCS_sec:prob_def}, we discuss the two problems in detail and propose a new modularity function named density modularity that can mitigate the two problems for community search, compared with the classic modularity.

\section{Problem Definition}\label{DMCS_sec:prob_def}
 
We proceed to define density modularity and the Density Modularity based Community Search (DMCS) problem. Note that all the proofs of lemmas and theorems can be found in appendix.

\begin{definition}
(\underline{\textcolor{black}{Density Modularity}}). \textcolor{black}{Given a weighted graph $G=(V,E)$ and a set of nodes $C$, the density modularity is defined as $DM(G, C) = \frac{1}{|C|} (w_C - \frac{d_C^2}{4w_G})$ 
where $w_C$ is the sum of internal edge weights of the community $C$, $d_C$ is the sum of the node weights of the community $C$, and $w_G$ is the sum of the edge weights of the graph $G$. Note that a node weight is the sum of adjacent edge weights. }
\end{definition}

\textcolor{black}{
As $G$ is always used, we simplify notation $DM(G, C)$ to $DM(C)$ when the context is clear. For unweighted graph, the definition is $DM(G,C)=  \frac{1}{2|C|}(2l_C - \frac{d_C^2 }{2|E|})$ where $l_C$ is the number of internal edges and $d_C$ is the sum of node degrees within the community $C$. }

\begin{example}\label{DMCS_example:mod_dm_example}

Let us reuse Figure~\ref{DMCS_fig:example} to illustrate how to compute the density modularity of community $A$ and $A\cup B$. Given $|E|=26$, $l_{A\cup B} = 14$,  $d_{A\cup B} = 28$, $l_{A} = 6$, and $d_{A}=14$, the density modularity of community $A$ and community $A \cup B$ is computed as follows.
\begin{itemize}[leftmargin=*]
    \item $DM(A)=\frac{1}{8}(12 - \frac{14^2}{52})=  1.028846$
    \item $DM(A\cup B)=\frac{1}{16}(28 - \frac{28^2}{52})= 0.8076923$
\end{itemize}
\end{example}

We defer the discussion on density modularity's benefits to Section~\ref{DMCS_sec:dm_property}, and we next give the definition of {\DMCS}.

\begin{pDefinition}
(\underline{DMCS}). 
Given a graph $G=(V,E)$ and a set of query nodes $Q$, the Density Modularity based Community Search (DMCS) aims to find a connected subgraph $G[C]$ (for $C\subseteq V$) that contains $Q$ such that $DM(G,C)$ is maximized.  
\end{pDefinition}

\subsection{Benefits and Hardness }\label{DMCS_sec:dm_property}

We discuss  benefits of density modularity. We first introduce the free-rider effect \cite{wu2015robust} problem and the resolution limit~\cite{fortunato2007resolution}, and then demonstrate that our density modularity suffers less from these problems compared with the classic modularity for the community search problem. We also give the hardness of {\DMCS}.

\begin{definition}(\underline{Free-rider effect}~\cite{wu2015robust,huang2015approximate})
Given query nodes $Q$, let $S$ be an identified community based on a goodness function $f$ and $S^*$ be a (local or global) optimum solution. We consider that the goodness function suffers from the free-rider effect if $f(S \cup S^*) \geq f(S)$.
\end{definition}

We next give a crucial lemma that proves our density modularity suffers less from the free-rider effect than the classic modularity. All the proofs are in Appendix.

\begin{lemma}\label{DMCS_lemma:DM_rlp_CM}
\textcolor{black}{Whenever density modularity suffers from the free-rider effect, the classic modularity suffers from the free-rider effect as well. }
\end{lemma}

\begin{remark}
\textcolor{black}{We note that the density modularity does not completely remove the free-rider effect. In Equation~\ref{DMCS_eq:mod_bound} of the Appendix~\ref{DMCS_appendix:FRE_lemma1}, 
if we consider that $|S_{int}|$ and $l_{int}$ are very small to be ignorable, $|S^*|=|S|$, $l_S = l_{S^*}$, and $d_S^2 > 2d_S d_{S^*} + d_{S^*}^2$, the equation does not hold, making the density modularity suffer from the free-rider effect. 
}
\end{remark}

\begin{figure}[t]
\centering
\includegraphics[width=0.99\linewidth]{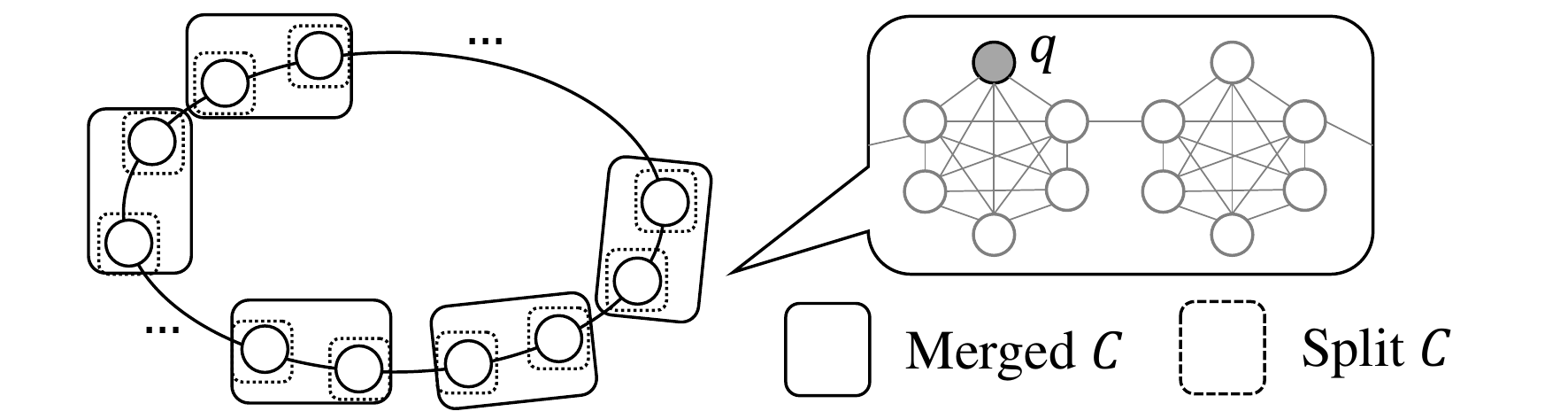}  
\vspace{-0.3cm}
\caption{Resolution limit problem}
\label{DMCS_fig:RL}
\end{figure}

\spara{Resolution limit problem.} 
We next present how our density modularity suffers less from the resolution limit problem compared with the classic graph modularity. 
One intuitive reason is that the term $|C|$ used in density modularity avoids finding large-sized communities. 
\textcolor{black}{
We next give the definition of the resolution limit problem. 
}
\begin{definition}\label{DMCS_def:RLP}
\textcolor{black}{
Given a graph $G$, query nodes $Q$, objective function $f$, a community constraint $C$, an identified subgraph $H$ satisfying $C$ and containing all the query nodes $Q$, and any subgraph $H'$ satisfying the constraint $C$ such that $G[H\cup H']$ is connected and $H\cap H'=\varnothing$, we say that the objective function suffers from the resolution limit problem for community search if there is a subgraph $H'$ such that $H\cup H'$ satisfies the constraint $C$ and $f(H\cup H') \geq f(H)$.
}
\end{definition}

\begin{remark}
\textcolor{black}{
Note that free-rider effect is  different from the resolution limit problem. The former considers the optimal solution $S^*$ to detect the effect. Particularly, $S^*$ always contains all the query nodes. In contrast, in the resolution limit problem, it is assumed that $H'$ is an independent community  of query nodes and there are no common nodes between the identified community $H$ and $H'$. Meanwhile, the connectivity of $G[H\cup H']$ must be guaranteed. }
\end{remark}

\begin{lemma}\label{DMCS_lemma:RL}
\textcolor{black}{
Whenever density modularity suffers from the resolution limit problem, the classic modularity suffers from the resolution limit problem as well. }
\end{lemma}
We also give an example as follows to explain why the density modularity mitigates the resolution limit issue.

\begin{example}
One famous example is a ring structure of cliques in a network~\cite{fortunato2007resolution,bettinelli2012algorithm}, as shown in Figure~\ref{DMCS_fig:RL}. The example contains $30$ cliques consisting of $6$ nodes. We compare two possible communities containing the query node $q$. One is called the merged community which merges two $6$-cliques. The other is called the split community, which is a $6$-clique. Intuitively, the latter result is better since each clique is fully connected internally and sparsely connected externally. However, we notice that the classic modularity prefers the former result and the reasons are as follows.  

Based on the classic modularity, the modularity scores for the merged community and split community (see the following) are respectively $0.0601$ and $0.0301$, indicating that the merged case is preferred. This illustrates that the classic modularity suffers from the resolution limit problem for the community search problem.

\begin{itemize}[leftmargin=*]
    \item merged community : $31/480 - (64/(2*480))^2 = 0.06013889$
    \item split community :  $15/480 - (32/(2*480))^2) = 0.03013889$
\end{itemize}

Based on our density modularity, we do not prefer a merged community as a result since our goodness function contains the size of the identified community. Thus, our density modularity mitigates the resolution limit problem compared with the classic modularity. 
\begin{itemize}[leftmargin=*]
    \item merged community : $31/12 - (64^2/(4*480*12))) =   2.405556$
    \item split community : $15/6 - (32^2/(4*480*6)))  = 2.411111$
\end{itemize}
\end{example}  

\begin{remark}
\textcolor{black}{
We note that the resolution limit problem is not a purely community-size related matter. We consider that an objective function suffers the resolution limit problem if it finds the sub-optimal merged subgraph as a result even if it consists of well-connected multiple communities with very sparse external edges. Thus, a community search model with size constraints cannot resolve the problem. In addition, selecting the suitable size of a community can still be challenging. 
}
\end{remark}

\noindent
{\bf Hardness.} We conclude this Section by giving the hardness of the DMCS problem with proof in Appendix.

\begin{theorem}\label{DMCS_theorem:NP}
The {DMCS} problem is NP-hard. 
\end{theorem}

\section{Algorithms}\label{DMCS_sec:alg}
\vspace{-0.1cm}
\textcolor{black}{
As the {\DMCS} problem is NP-Hard, there exist no polynomial time algorithms to answer {\DMCS} exactly unless P=NP. We develop  effective and efficient heuristic algorithms.}
We introduce an algorithmic framework to solve {\DMCS}, based on which we design two
algorithms: (1) Non-articulation  Cancellation Algorithm ({\NCA});
(2) Fast Peeling Algorithm (\FPA) with a better complexity. 
For ease of presentation, we first assume the query has a single query node, and relax the assumption in Section~\ref{DMCS_sec:multipleQueryNodes}.

\subsection{Algorithm Framework}\label{DMCS_sec:framework}

Our algorithmic framework (Algorithm~\ref{DMCS_alg:overall_framework}) follows a \textit{top-down greedy} fashion, which (a.k.a peeling-based approach) is widely used for the community search problem~\cite{fang2020effective,wang2021efficient,yang2020effective}. \textit{Top-down} implies that we iteratively remove nodes based on some criteria (Line 3). \textit{Greedy} indicates that we choose the best node to be removed so as to maximize the density modularity iteratively (Line 4). 
A node is \textit{removable} if it is not a query node and removing it does not disconnect the remaining graph since our resultant community must be connected.  
There are two key functions in the framework: 1) checking the existence of removable nodes (i.e., finding removable nodes) in Line 3; and 2) finding the best node to be removed in Line 4. 

The novelty of our algorithms comes from how  we design the two functions. For each function, we design two new solutions, which are  summarized in Figure~\ref{DMCS_fig:alg_summary}.  
Based on the designed functions, we have two new algorithms with different complexity: (1) Non-articulation Cancelling Algorithm; (2) Fast Peeling Algorithm. The high level idea of the two algorithms is as follows. 
Note that we also empirically evaluate the other combinations of the designed functions in Section~\ref{DMCS_sec:alg_variation}.

\begin{algorithm}[t]
\SetKw{break}{break}
\SetKw{updated}{updated}
\SetKw{return}{return}
\SetKw{false}{false}
\SetKw{true}{true}
\SetKwData{C}{C}
\SetKwData{dist}{dist}
\SetKwFunction{max}{max}
\SetKwInOut{Input}{input}
\SetKwInOut{Output}{output}
\Input{Graph $G=(V,E)$, query nodes $Q$, goodness function $M$}
\Output{A community \C}
$i \leftarrow 1$\;
$G^i \leftarrow G$\;
\While{ \text{$G^i$ has removable nodes} }{
    $v \leftarrow$ best removable node  maximizing $M( G^i-\{v\})$\;
    $G^i \leftarrow G^i-\{v\}$\;
    $i++$\; 
}
\return  $\argmax_{G''\in \{G^1, G^2, \cdots G^{i-1}\}} M(G, G'')$  \;
\caption{Overall Framework}
\label{DMCS_alg:overall_framework}
\end{algorithm} 

\spara{Non-Articulation Cancelling Algorithm ({\NCA}).} The high-level idea of {\NCA} is to iteratively find removable nodes and then remove one of them which can maximize the density modularity in a greedy manner. To compute the removable nodes, 
{\NCA} finds every \emph{non-articulation node}, which is a node whose removal does not disconnect the graph. 
We say a non-articulation node is {\it removable} if it is not a query node. For every iteration, we aim to find a removable node such that removing it retains the maximal density modularity.

To sum up, the overall procedure of {\NCA} is as follows: (1) finding all the non-articulation nodes in the current graph; (2)  computing the density modularity after removing each non-articulation node (excluding the query nodes); (3) selecting a node such that removing it can gain the largest possible resultant density modularity; and (4) going back to step (1). This process is repeated until there is no non-articulation node in the current subgraph. Finally, we return a subgraph having the largest density modularity among all the intermediate subgraphs. {\NCA} is further presented in Section~\ref{DMCS_sec:NCA}.

\spara{Fast Peeling Algorithm (\FPA).} 
The high-level idea of {\FPA} is motivated by the observation that real-world social networks are typically scale-free~\cite{barabasi2003scale,barabasi2009scale}, leading to communities with small diameters. In other words, the community itself is a small world where any two nodes in it are not distant. With this observation, we design the {\FPA} that iteratively removes nodes that are located farthest from the query node. Obviously, the farthest node from the query node is removable (i.e., does not make the graph unconnected) when there is a single query node. 
To find the best removable node, we define a new goodness function named \textit{density ratio} which is a lightweight version of the density modularity. Using the density ratio allows the algorithm to only update the density modularity of the nodes that directly connect to the removed node in the next iteration, leading to much higher efficiency.  More details on {\FPA} are presented in Section~\ref{DMCS_sec:PeelingAlgorithm}.
\vspace{-0.1cm}
\begin{figure}[ht]
\centering
\includegraphics[width=0.99\linewidth]{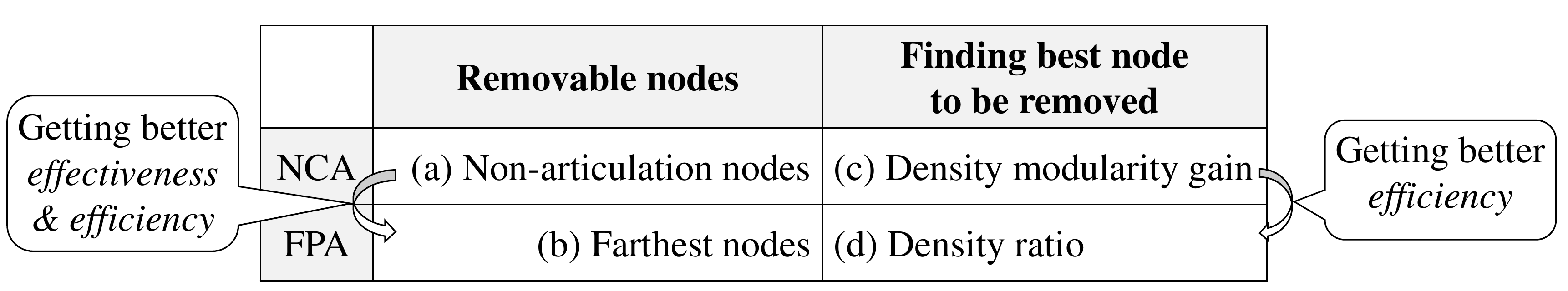}
\vspace{-0.4cm}
\caption{Key functions of proposed algorithms}
\label{DMCS_fig:alg_summary}
\vspace{-0.5cm}
\end{figure}

\subsection{Function 1 : Computing Removable Nodes}\label{DMCS_sec:sec5_2}
We design two new solutions to the problem of  how to compute removable nodes, which are based on non-articulation nodes and shortest path distances, respectively.
\subsubsection{Removable nodes based on non-articulation nodes.}\label{DMCS_sec:non_articulation}
In this solution, we identify a non-articulation node as a candidate removable node.
Instead of directly finding non-articulation nodes, we compute its complementary set that contains articulation nodes. A node $v$ in $V$ is called an articulation node of a graph $G=(V,E)$ if $G[V\setminus v]$ is unconnected.  
We compute articulation nodes based on DFS-tree~\cite{hopcroft1973algorithm}. 
In DFS-tree, a node $x$ is an ancestor of a node $y$ if the node $y$ is visited by the node $x$ through the Depth-First Search (DFS). The node $x$ in DFS-tree is an articulation node if (1) $x$ is a root node of DFS-tree and has at least two children nodes or (2) $x$ is not a root node of DFS-tree and has a child $y$ such that no node in sub-tree rooted at node $y$ has an edge connected to one of the $x$'s ancestors in DFS-tree.  
Note that after removing a node from the current subgraph, 
a non-articulation node can become an articulation node and vice versa. Therefore, it needs to examine whether a node is a non-articulation node in every iteration.

\subsubsection{Removable nodes based on shortest path distance.}
Computing all removable nodes is time-consuming. To alleviate this issue, we propose a heuristic approach to find a set of removable nodes. 
We first consider a set of farthest nodes from the query node as the removable nodes. Suppose the diameter of the social network is $D$, then the distances of other nodes to the query node can only fall into one of $D$ groups $S_1$, $S_2$, $\ldots,$ $S_D$, where $S_i$ contains all the nodes, each having a shortest path distance $i$ from the query node. Then, we can always remove the nodes in $S_D$ first because removing any of them will not disconnect the graph. In other words, the nodes in $S_D$ are all non-articulation nodes. After removing every node in $S_D$, all nodes in $S_{D-1}$ become non-articulation nodes, and we then consider removing the nodes in $S_{D-1}$, so on so forth.

The idea is more effective when the graph has a smaller diameter because it generates smaller node groups $S_i (1\leq i\leq D)$ and each group has more candidate removable nodes to select. We investigate the diameters of datasets used in our experiments, and we find that most communities have very small diameters. 
Figure~\ref{DMCS_fig:diameter_realworld} shows the results on two example datasets: for DBLP dataset, around $80\%$ communities have their diameters at most $4$; for Youtube dataset, about $94\%$ communities have diameters at most $4$. The results on other datasets are similar. Therefore, using shortest path distance to select removable nodes is effective for social networks.

\begin{figure}[t]
\centering
\includegraphics[width=0.99\linewidth]{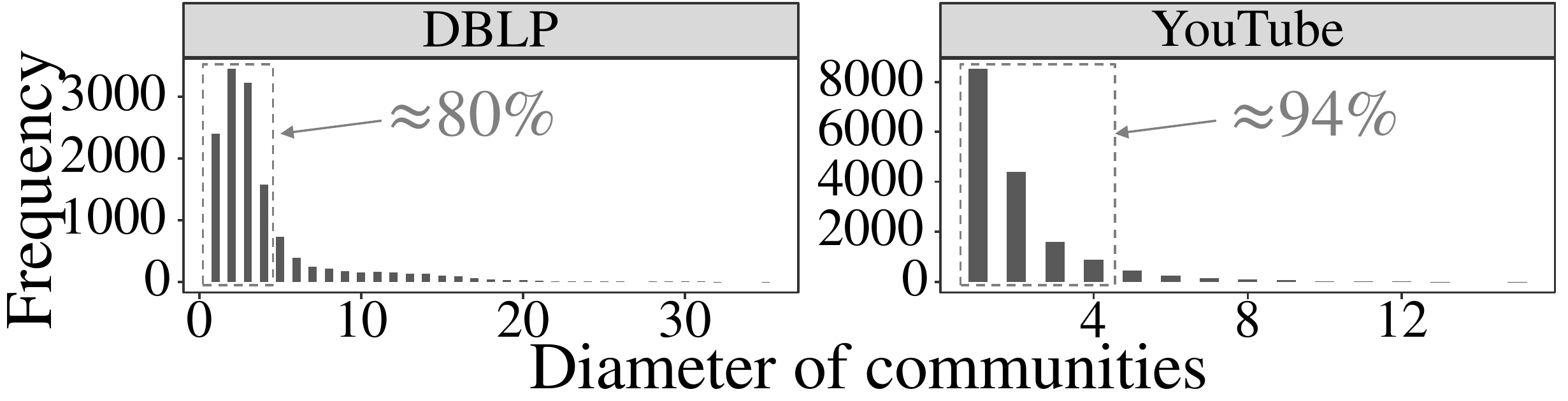}
\vspace{-0.3cm}
\caption{Frequency of community diameter}
\label{DMCS_fig:diameter_realworld}
\end{figure}

\subsection{\mbox{Function 2: Finding the Best Removable Node}}
We design two solutions to finding the best removable node, which are based on the density modularity and density ratio. 
\subsubsection{Density modularity.}\label{DMCS_sec:density_modularity_alg}  
The best removable node is the node whose removal retains the largest density modularity. Intuitively, we define the {\it updated density modularity} as the density modularity after removing a specific node. For every iteration, we choose the best node 
that has the largest updated density modularity among all the candidate removable nodes. The updated density modularity can be defined 
as the difference between the density modularity before and after removing a specific node $v$ (see Definition~\ref{DMCS_def:updateddensitymodularity}).
\begin{definition}\label{DMCS_def:updateddensitymodularity}
(\underline{Updated density modularity}).\\
\textcolor{black}{
Given a graph $S$ and a node $v$, the updated density modularity is the density modularity after removing a specific node $v$ and its value is  
$\frac{l_{S} - k_{v, S}}{|S|-1} - \frac{(d_{S} - d_v)^2}{4|E|(|S|-1)}$, 
where $k_{v, S}$ is the number of edges from the node $v$ to the subgraph $S$, $d_S$ is the sum of node degrees in $S$, and $d_v$ is the degree of the node $v$.
}
\end{definition}
In the updated density modularity, $\frac{1}{|S|-1}$, $l_S$, and $d_S^2$ are fixed when comparing two candidate removable nodes, and therefore can be ignored. Removing these fixed terms gives us the following concept called {\it density modularity gain}. We use the simpler density modularity gain, which has the same effect as the updated density modularity, in selecting the best removable node.  
\begin{definition}\label{DMCS_def:densitymodularityGainDef}
(\underline{Density modularity gain $\Lambda$}).\\
Given a graph $S$ and a node $v$, the density modularity gain $\Lambda$ is 
$\Lambda_v^S = -4|E|k_{v,S} + 2d_Sd_v - d_v^2$. 
\end{definition}

\noindent
{\bf Remarks on Density Modularity.} In the top-down greedy algorithm, we are interested in a {\it stable} goodness function to improve the efficiency. 
In particular, we consider that a function is stable if removing any node $u$ does not affect the function values of the other nodes that are not connected to $u$. If the goodness function is stable in our problem, then we only need to update the neighbour nodes of the removing node in every iteration. This can significantly boost the efficiency since we do not need to recompute the values of all candidate removable nodes. Unfortunately, the density modularity gain is unstable, and we  introduce a new goodness function  in Section~\ref{sec: density ratio}.

\begin{lemma}\label{DMCS_lemma:notstable}
The density modularity gain is not stable.
\end{lemma}

\begin{proof}
We notice that if a node $u$ is removed, the $\Lambda$ values of other nodes which are not connected to the node $u$ decrease since the term $d_S$ decreases. 
\end{proof}

\subsubsection{Density ratio.} \label{sec: density ratio}
We present a new goodness function named the density ratio (see Definition~\ref{DMCS_def:densityratio}) that is stable.

\begin{definition}\label{DMCS_def:densityratio}
(\underline{Density ratio $\Theta$}).\\
Given a graph $G$, a subgraph $S$ of $G$ and a node $v$ in $S$, the density ratio of the node $v$ is $\Theta_v^{S} = \frac{d_v}{k_{v, S}}$ where $k_{v, S}$ is the number of edges from $v$ to $S$ and $d_v$ is the degree of the nodes in $G$. 
\end{definition}

The intuition of density ratio is that it is a lightweight version of the density modularity gain.
Let us recall the density modularity gain $ \Lambda_v^S = -4|E|k_{v,S} + 2d_Sd_v - d_v^2 $. The value $\Lambda_v^S$ can be maximized when $d_v$ is large and $k_{v,S}$ is small. Therefore, we define the density ratio as the ratio between $d_v$ and $k_{v,S}$. Such a refactoring leads to the following interesting property.

\begin{lemma}
The density ratio function is stable. 
\end{lemma}
\begin{proof}
\textcolor{black}{
Suppose that we remove a node $u$. Note that the variable $d_v$ for all the nodes $v\in S$ is fixed since the original degree is not changed. 
Then, only $k_{v,S}$ decreases if node $u$ is connected to node $v$; $k_{v,S}$ has no change otherwise. 
It implies that the density ratios of the nodes that are not connected to node $u$ do not change. Therefore, the density ratio function is stable.
}
\end{proof}

The stability of the density modularity can improve the efficiency since it only requires to update the density ratios of the neighbour nodes of the last removed node. Particularly, in Line 4 of the overall framework (Algorithm~\ref{DMCS_alg:overall_framework}), in every iteration we find the best removable node to maximize the goodness function. Applying the density ratio as the goodness function $M$, we only need to compute $M( G^i-\{v\})$ for $v$ that is a neighbor of the last removed node. Note that computing the density ratios of all the nodes at the initial stage is required when we use the density ratio.  

\begin{figure}[t]
\centering
\includegraphics[width=0.99\linewidth]{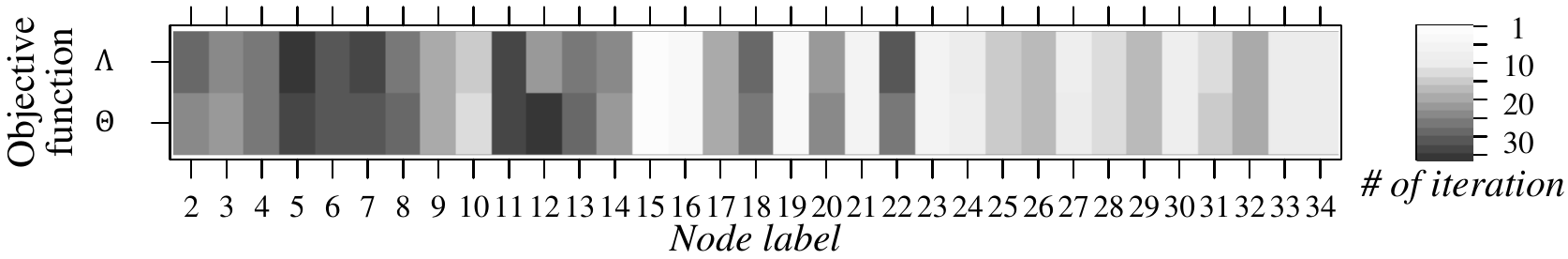}
\vspace{-0.4cm}
\caption{Difference of the update order}
\label{DMCS_fig:heating}
\end{figure}

Another interesting property of the density ratio is that it only mildly impacts the node removal order compared with using the density modularity gain. In other words, using the density ratio retains effectiveness of the density modularity gain, and also significantly improves the efficiency. To investigate how similar the two goodness functions
($\Lambda$ and $\Theta$) are, we plot a heatmap in
Figure~\ref{DMCS_fig:heating} to illustrate the difference of the node removal orders on the Karate network~\cite{zachary1977information}. The brighter color indicates that the node is removed earlier, and the darker color indicates that the node is removed later. We observe that the two goodness functions have very similar removing orders since the density ratio is closely related to the density modularity gain.

\subsection{Non-articulation Cancellation Algorithm}\label{DMCS_sec:NCA}

We present the Non-articulation Cancellation Algorithm, dubbed as {\NCA}, that uses non-articulation nodes to find removable nodes, and uses the density modularity gain to choose the best node to be removed. In other words, {\NCA} uses (a) and (c) in Figure~\ref{DMCS_fig:alg_summary}. 
The procedure of {\NCA} is as follows. 
At the initial stage, we check whether all the nodes in $Q$ are in the same connected component. If not, the algorithm terminates. Next, for all the nodes, we compute the minimum shortest path distance from the query nodes. 
Next, it identifies removable nodes by computing articulation nodes (See Section~\ref{DMCS_sec:non_articulation}), then removes the node with the largest density modularity gain. If two nodes have the same density modularity gain, we keep the node that is closely located to the query nodes. 
After removing the node, we update the solution to the currently obtained community if the latter has larger density modularity. This process is repeated until there are no more removable nodes in the current subgraph. Finally, we return the current solution which has the largest density modularity among all the intermediate subgraphs.

\spara{Time  complexity.}
The complexity of each component of {\NCA} is
\begin{itemize}[leftmargin=*]
    \item It takes $O(|V|+|E|)$ to compute all the  non-articulation nodes. This computation is repeated at most $|V|$ times. 
    \item It takes $O(|V|)$ to compute $\Lambda$ for all the nodes in the current subgraph. This computation is repeated at most $|V|$ times.
    \item Computing the minimum shortest distance from the query nodes takes $O(|Q|(|E| + |V|\log{|V|}))$. 
\end{itemize}
Therefore, the total time complexity of {\NCA} is $O(|V|(|V|+|E|)+|Q|(|E| + |V|\log{|V|}))$ since $|V|^2$  less than $|V|(|V|+|E|)$. It indicates that the main bottleneck of {\NCA} is to compute non-articulation nodes for every iteration. In the next subsection, we discuss how the Fast Peeling Algorithm resolves this issue. 

\spara{Limitation.}
One major limitation of {\NCA} is that it is sensitive to the graph structure. Figure~\ref{DMCS_fig:local_optimum} is the result of {\NCA} in a synthetic network. 
We observe that {\NCA} returns a community (marked by ``resultant community'')  that contains two connected communities even if $Q$ has a more closely related community (marked by ``Ground-truth community''). 
{\NCA} returns such a result since nearby connected community is more densely connected internally and sparsely connected externally. Thus, {\NCA} 
will remove a set of nodes in the ground-truth community 
because the benefits (density modularity) to keep the nodes in right-side communities are larger than keeping the nodes in the ground-truth community. 
To address this issue,
we present {\FPA} (in Section~\ref{DMCS_sec:PeelingAlgorithm}) which employs a distance-based removing approach presented in Section~\ref{DMCS_sec:sec5_2}.

\begin{figure}[t]
\centering
\includegraphics[width=0.99\linewidth]{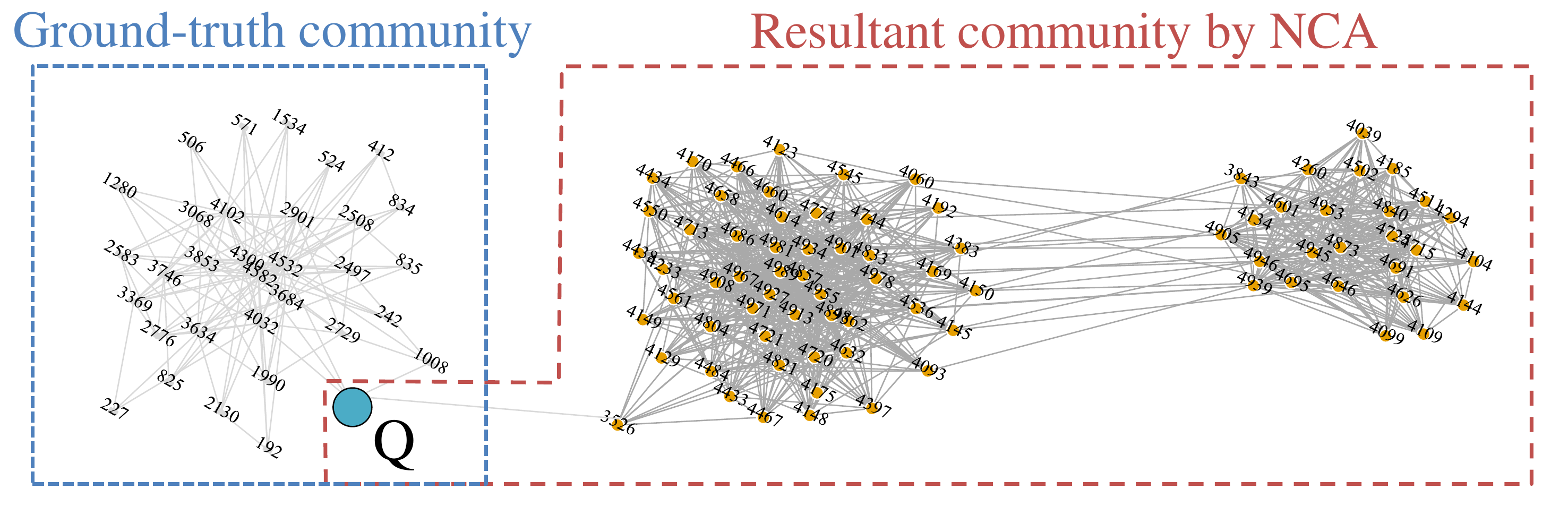}
\vspace{-0.4cm}
\caption{Local optimum solution of {\NCA}}
\label{DMCS_fig:local_optimum}
\end{figure}

\subsection{Fast Peeling Algorithm}\label{DMCS_sec:PeelingAlgorithm}

We next present the Fast Peeling Algorithm, termed as {\FPA}, that incorporates the shortest path distance to find removable nodes, as well as using the density ratio to find the best node to be removed. In other words, {\FPA} uses (b) and (d) in Figure~\ref{DMCS_fig:alg_summary}. 
The procedure of {\FPA} is described in 
Algorithm~\ref{DMCS_alg:FastPeelingAlg}. 
In Line 1--2, we initialize variables. 
In Line 3, we compute the shortest path distances from the query node. We discuss how to compute the distance when there are multiple query nodes in Section~\ref{DMCS_sec:multipleQueryNodes}. 
In Line 4, we sort all the nodes based on the distances. 
In Line 6, we find a set of nodes which are the farthest from the query node. 
In Line 7--12 we find the node $u$ that has the largest density ratio and removes $u$ from the current subgraph $S$, and update the density ratio of the neighbour nodes of $u$. 
In Line 13--14, we update the solution if the density modularity of the currently obtained subgraph has larger density modularity than the current solution. 
If the current subgraph has no more removable node, we return the solution $C$ as a result in Line 15.

\begin{algorithm}[t]
\SetKw{break}{break}
\SetKw{updated}{updated}
\SetKw{return}{return}
\SetKw{false}{false}
\SetKw{true}{true}
\SetKwData{dist}{dist}
\SetKwFunction{CC}{connectedComp}
\SetKwFunction{maxDelta}{max$\Delta$}
\SetKwFunction{max}{max$\Theta$}
\SetKwFunction{getNonArticulationNodes}{getNonArticulationNodes}
\SetKwFunction{computeDist}{computeDist}
\SetKwFunction{sortByDist}{sortByDist}
\SetKwFunction{degreeToNodes}{degreeToNodes}
\SetKwFunction{findBest}{FindBestRemovableNode}
\SetKwFunction{nodesInMaxDist}{nodesInMaxDist}
\SetKwInOut{Input}{input}
\SetKwInOut{Output}{output}
\Input{Graph $G=(V,E)$, query nodes $Q$}
\Output{A community $C$}
$S \leftarrow$ \CC{$G$, $Q$}\;
$C \leftarrow S$\;
\computeDist{$S$, $Q$}\;
$S \leftarrow $\sortByDist{$S$}\;
\While{maximum node distance in $S$ is not $0$}{
    $Cand \leftarrow $ \nodesInMaxDist{$S$}\;
    \While{$|Cand| \neq 0$}{
        $u \leftarrow \argmax_{v\in Cand}\Theta_v$\;
        $S \leftarrow S\setminus u$\;
        $Cand \leftarrow Cand\setminus u$\;
        \For{$v \in N(u, S)$}{
            update $\Theta_v^S$\;
        }
        \If{$DM(S) \geq DM(C)$}{
            $C \leftarrow S$\;
        }
    }
}
\return $C$ \;
\caption{Fast Peeling algorithm}
\label{DMCS_alg:FastPeelingAlg}
\end{algorithm}

\spara{Time complexity.} The running time of each component in {\FPA} is as follows. 
\begin{itemize}[leftmargin=*]
    \item It takes $O(|E|+|V|\log{|V|})$ to handle multiple query nodes.
    \item Computing the shortest distance from the query nodes takes $O(|E|+|V|\log{|V|})$ by using Dijkstra algorithm\footnote{When we construct a hyper graph by merging the query nodes, it can be achieved}.  
    \item It takes $O(|V| \log{|V|})$ to sort the nodes based on $\Theta$ value. 
    \item It takes $O(|E|\log{|V|})$ to remove a set of nodes. In our recursive removing process, if we remove a node, the density ratios of its neighbour nodes are updated. Since we use binary search to maintain the order of the nodes based on the density ratio, it takes $O(|E|\log{|V|})$ where adding a node takes $O(\log{|V|})$ and the maximum number of the node to be added is $|E|$.
\end{itemize}
Therefore, the total running time of {\FPA} is $O(|V||Q|^2+diam(G)+(|E|+|V|)\log{|V|})$, and it is log-linear time to the graph size. 

\subsection{Handling Multiple Query Nodes.}\label{DMCS_sec:multipleQueryNodes}

In {\NCA}, we do not need to do additional tasks to handle multiple query nodes since removing non-articulation nodes guarantees the connectivity of the remaining graph. For {\FPA}, when the query node set $Q$ contains multiple query nodes, we compute the shortest path distance as follows. For each node $v\in V$, we let $dist(v)$ be the minimum shortest path distance from the query nodes, i.e., $dist(v) = \min{(dist(q, v))}, \forall q\in Q$. 

Furthermore, when there are multiple query nodes, we cannot guarantee the connectivity of the remaining graph after removing a node. To address this issue, we find a connected component containing all the query nodes and then set the connected component as the query nodes to guarantees that removing any farthest node does not disconnect the remaining graph. 
Computing a minimum subgraph containing the query nodes is the same as finding a solution of the Steiner-tree problem~\cite{hwang1992steiner}. In this paper, we use a simple approach as follows: 
(1) we firstly randomly choose a query node $q\in Q$; 
(2) we compute the shortest paths from the query node $q$ to all the other nodes; 
(3) we pick a set of the shortest paths $T$ whose endpoints belong to query nodes; 
(4) we merge the shortest paths $T$; and 
(5) we return the subgraph induced by the nodes in $T$. This procedure takes  $O(|E|+|V|\log{|V|})$ time.

\subsection{Layer-based pruning strategy}
\textcolor{black}{To improve the efficiency of {\FPA}, we present a layer-based pruning strategy. The high-level idea of the pruning strategy is inspired by the observation of the ground-truth communities on real-world networks in Section~\ref{DMCS_sec:sec5_2}. 
It implies that from the query nodes, we may not need to iteratively remove the nodes that are far from the query nodes. 
Therefore, based on the minimum shortest path distance from the query nodes, we construct a set of layers: $L_1, L_2, \cdots L_g$. Note that a set of nodes in the same layer have the same minimum shortest path from the query nodes, and the number of layers is small as a graph diameter is normally small~\cite{watts2004six,watts1998collective}.
Next, by iteratively removing the outermost layer, we get a set of subgraphs, and 
we select the subgraph having the largest density modularity. Finally, we apply the node-removing process to the outermost layer of the selected subgraph  to find a solution.
We observe that in DBLP dataset, when the layer-based pruning strategy is applied to {\FPA}, running time gets faster up to $300$ times than {\FPA} without the pruning strategy.
Figure~\ref{DMCS_fig:layer} presents the procedure of the layer-based pruning strategy. 
}

\begin{figure}[ht]
\centering
\includegraphics[width=0.99\linewidth]{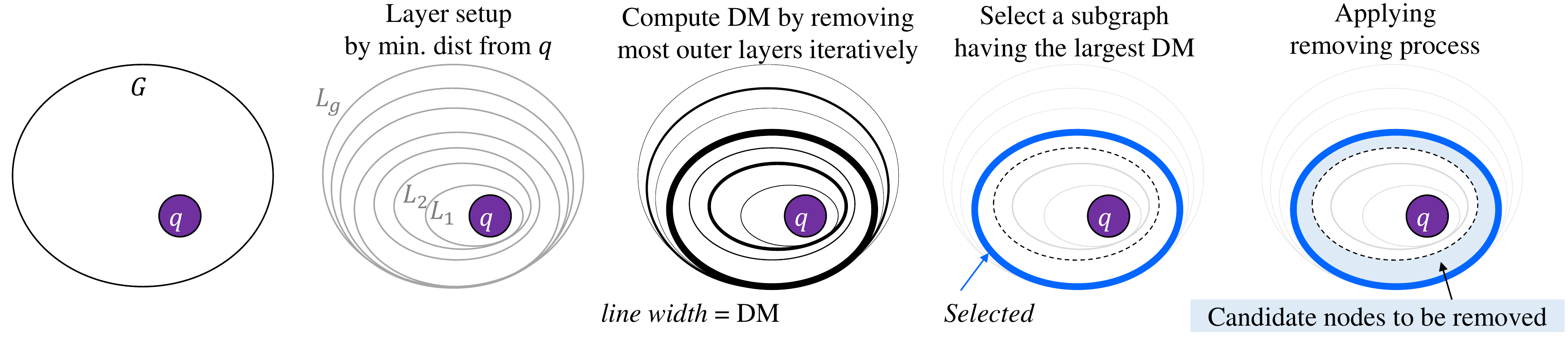}
\vspace{-0.3cm}
\caption{Layer-based pruning strategy}
\label{DMCS_fig:layer}
\end{figure}


\section{Experiments}
\label{DMCS_sec:experiments}

We evaluate our algorithms over real-world and synthetic networks. All experiments were conducted on a machine with CentOS 8 with 128GB memory and 2.50GHz Xeon CPU E5-4627 v4. 
\subsection{Experimental Setting}
\spara{Real-world Dataset.} We use 7 real-world datasets with ground-truth community information, whose statistics are reported in Table~\ref{DMCS_tab:dataset-realworld}.
We denote $|C|$ as the number of communities and `overlap' indicates whether the community membership is overlapping.

\begin{table}[ht]
\vspace{-0.3cm}
\caption{Real-world datasets}
\label{DMCS_tab:dataset-realworld}
\vspace{-0.3cm}
\centering
\begin{tabular}{c|c|c|c|c}
\hline
            & $|V|$         & $|E|$          & $|C|$        & overlap \\ \hline \hline

Dolphin~\cite{lusseau2003bottlenose}     & 62        & 159        & 2               & \xmark           \\ \hline
Karate~\cite{zachary1977information}      & 34        & 78         & 2             & \xmark           \\ \hline
Polblogs~\cite{adamic2005political}    & 1,224      & 16,718      & 2           & \xmark           \\ \hline
Mexican~\cite{gil1996political}    & 35      & 117      & 2           & \xmark           \\ \hline
\hline
DBLP~\cite{yang2015defining}        & 317,080    & 1,049,866    & 13,477      & \cmark           \\ \hline
Youtube~\cite{yang2015defining}     & 1,134,890   & 2,987,624    & 8,385    & \cmark           \\ \hline
\textcolor{black}{Livejournal}~\cite{yang2015defining}     & \textcolor{black}{3,997,962}   & \textcolor{black}{34,681,189}    & \textcolor{black}{287,512}    & \cmark           \\ \hline
\hline
\end{tabular}
\end{table} 

\spara{Synthetic Dataset.} We use the LFR benchmark dataset~\cite{lancichinetti2008benchmark} which is a synthetic dataset to evaluate the accuracy of community detection algorithms. The parameters of the synthetic network are described in Table~\ref{DMCS_tab:syn} and default parameters are underlined. 
\begin{table}[ht]
\vspace{-0.1cm}
\caption{Synthetic network configuration}
\label{DMCS_tab:syn}
\vspace{-0.3cm}
\centering
\begin{threeparttable}
\begin{tabular}{c|c|c}
\hline
\textbf{Var} & \textbf{Values}   & \textbf{Description}    \\ \hline \hline
$|V|$                & \textbf{5,000}           & the number of nodes             \\ \hline
$d_{avg}$                & 20,30,\underline{\textbf{40}}, 50  & average degree          \\ \hline
$d_{max}$            & \underline{\textbf{200}},  300,400,500             & maximum degree          \\ \hline
$\mu$              & 0.2,0.3,\underline{\textbf{0.4}} & mixing parameter\tnote{*}       \\ \hline
min $C$            & \textbf{20}              & minimum community sizes \\ \hline
max $C$            & \textbf{1,000}           & maximum community sizes \\ \hline \hline
\end{tabular}
\begin{tablenotes}
  \item[*] the ratio of inter to intra-community edges.
\end{tablenotes}
\end{threeparttable}
\end{table}

\spara{Ground-truth Communities.} 
For real-world networks, ground-truth communities are formed as follows: 
(1) Dolphin~\cite{lusseau2003bottlenose} : Communities of nodes (dolphins) are divided into the male ones and female ones;
(2) Karate~\cite{zachary1977information} : A conflict between two club members leads to nodes (club members) to divide into two groups. Each group indicates a community;
(3) Polblogs~\cite{adamic2005political} : The nodes (blogs) are separated based on their political orientation;
(4) Mexican~\cite{gil1996political} : The nodes (politicians) are grouped by their positions (civil or military);
(5) DBLP~\cite{yang2015defining} : The nodes (authors) who have published in a specific journal or conference form a community; and 
\textcolor{black}{
(6) Youtube and Livejournal~\cite{yang2015defining} : Social communities of nodes (users) are defined as user-defined groups. 
}

\spara{Algorithms.} \textcolor{black}{We compared our algorithms ({\NCA} and {\FPA}) with several baseline algorithms. We only report the results when the baseline algorithms return a result within 24 hours.}

\begin{itemize}[leftmargin=*]
    \item clique based community search~\cite{yuan2017index} ({\clique})
    \item $k$-core based community search~\cite{sozio2010community} ({\kc})
    \item $k$-truss based community search~\cite{huang2014querying} ({\kt})
    \item \textcolor{black}{$k$-edge connected component~\cite{chang2015index} (\kecc)}
    \item \textcolor{black}{CNM's agglomerative algorithm~\cite{clauset2004finding} (\CNM)}
    \item \textcolor{black}{GN's divisive algorithm~\cite{girvan2002community} (\GN)}
    \item \textcolor{black}{Luo's modularity~\cite{icwi2006paper} based greedy algorithm~\cite{luo2008exploring} (\icwi)}
    \item \textcolor{black}{Huang's basic algorithm~\cite{huang2015approximate} (\huang)}
    \item \textcolor{black}{Greedy node deletion algorithm~\cite{wu2015robust} (\wu)}
    \item \textcolor{black}{highest core-based community search (\highcore)}
    \item \textcolor{black}{highest truss-based community search (\hightruss)}
    \item Non-articulation Cancellation Algorithm ({\NCA})
    \item Fast Peeling Algorithm ({\FPA})
\end{itemize}
\textcolor{black}{For {\wu}, we implement the greedy algorithm \cite{wu2015robust} and set its parameter $\eta=0.5$. 
For {\GN}, it  
iteratively deletes a set of edges based on the betweenness centrality until no edges can be removed. For {\CNM}, it iteratively merges communities until there remains a single community. For both approaches, among the intermediate subgraphs containing all the query nodes, we pick the community which has the largest density modularity. 
For {\huang}, we implement the algorithm with a $2$-approximation ratio. 
}

\spara{Parameter Setting.} For {\kc} and {\kecc}, we set $k=3$ by default. For {\kt}, we set $k=4$ by default since $(k+1)$-truss contains $k$-core.  We report the experimental result by varying the parameter $k$ in Section~\ref{DMCS_sec:changeVarK}. For {\wu}, we set $\eta=0.5$. 

\spara{Query Sets.}
For all the networks, we pick $20$ sets ($10$ sets for small-sized datasets) of query nodes from the result of $(k+1)$-truss so that the query nodes are more likely to be located in a meaningful community. If there are over $20$ ground-truth communities, we randomly choose $20$ communities and then randomly pick a query set from each community. If there are fewer than $20$ ground-truth communities, we pick query sets such that they are most equally generated from each community.

\begin{figure*}[t]
\centering
\includegraphics[width=0.9\linewidth]{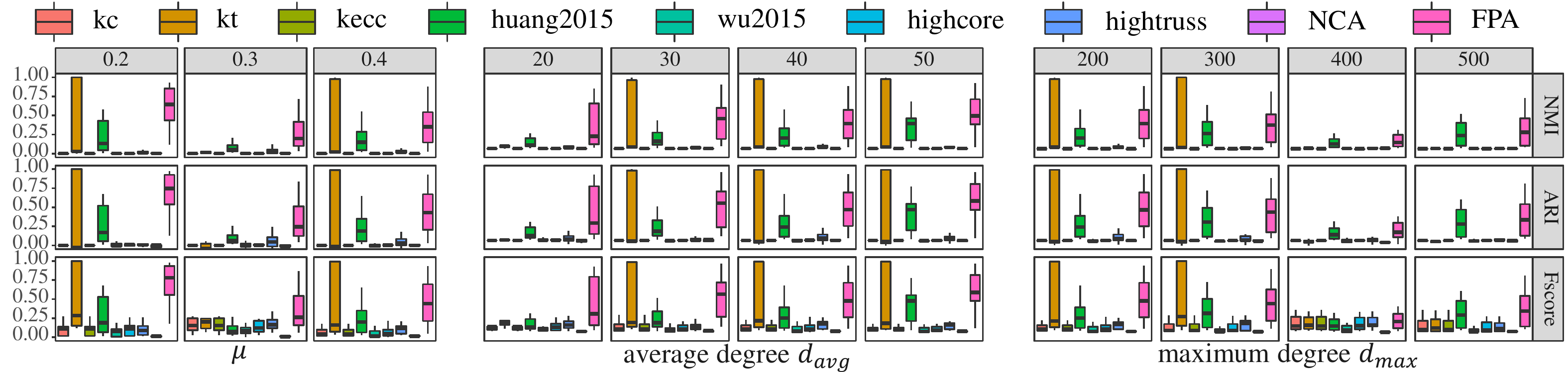}
\vspace{-0.3cm}
\caption{Effectiveness on benchmark networks}
\label{DMCS_fig:syn_all}
\vspace{-0.3cm}
\end{figure*}
\begin{figure*}[t]
\centering
\includegraphics[width=0.9\linewidth]{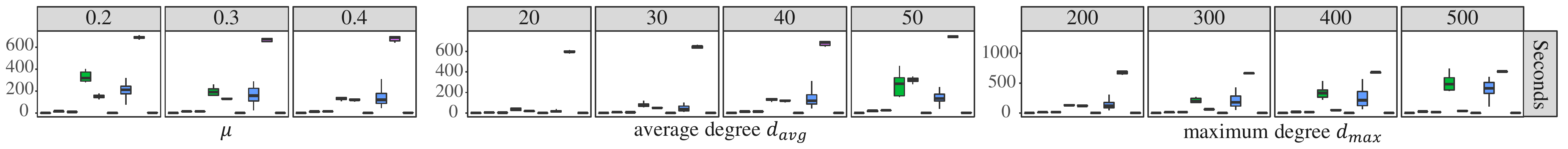}
\vspace{-0.3cm}
\caption{Efficiency on benchmark networks (same legend with Figure~\ref{DMCS_fig:syn_all})}
\label{DMCS_fig:syn_second}
\vspace{-0.3cm}
\end{figure*}

\spara{Evaluation Metric.} 
We employ the widely used metrics for the community detection problem with ground-truth information. These include the Normalized Mutual Information (NMI)~\cite{danon2005comparing} and Adjusted Rand Index (ARI)~\cite{hubert1985comparing}. We also use Fscore~\cite{van1979information} for evaluating the accuracy. 
Note that higher NMI, ARI, and Fscore indicate that the identified community and ground-truth community match better. 
To measure the accuracy in community search, we convert the community search problem into a binary classification problem. We consider a ground-truth community containing the query nodes as a true label in the binary classification problem. After we identify a community containing the query nodes, we compare our result with the true label. If there are multiple query nodes and they are not in the same ground-truth community,  this evaluation is not applicable. 
Note that for the binary classification tasks (the community search can be considered as binary classification task), Fscore measure returns overoptimistic inflated results~\cite{chicco2020advantages} since it cannot capture the ratio between positive and negative elements.

\subsection{Performance on Synthetic Graphs}

Figures~\ref{DMCS_fig:syn_all} and \ref{DMCS_fig:syn_second} show  effectiveness and efficiency results, respectively, on the benchmark dataset~\cite{lancichinetti2008benchmark}, where we vary the average degree $d_{avg}$, maximum degree $d_{max}$, and mixing parameter $\mu$ in the benchmark dataset.
We observe that {\FPA} is the fastest algorithm and returns high-quality result compared with other baseline algorithms and  {\NCA} has a scalability issue.  

\textcolor{black}{
Figure~\ref{DMCS_fig:syn_all} shows effectiveness of the returned communities. %
In general, {\FPA} and {\huang} have better accuracy than the other baseline algorithms and {\NCA}. 
This is because ground-truth communities normally have small diameters and {\FPA} considers internal and external edges simultaneously. The baseline algorithms including {\kc}, {\kt}, {\highcore}, {\hightruss}, and {\kecc} tend to return large communities, and hence often have lower accuracy. {\NCA} returns undesirable results due to the free-rider effect. We observe that {\wu} returns a community containing the nodes which are close to the query nodes. However, it does not indicate 
an accurate solution. 
We proceed to analyze the effect of average degree $d_{avg}$, maximum degree $d_{max}$, and mixing parameter $\mu$ on accuracy. 
}
First, the left 3 figures show the result when we vary the mixing parameter $\mu$. The $\mu$ is the fraction of edges that are between different communities. It implies that the larger $\mu$ is, the less detectable the community is. We observe that when $\mu$ becomes larger, the accuracy decreases since searching a community becomes harder because of the increase of the fraction of edges between communities.
Next, the middle 4 figures show that the accuracy of all the algorithms are not sensitive to the average degree $d_{avg}$. This is  because the fraction of edges that are between different communities is fixed. 
Finally, the right 4 figures show that  when the maximum degree $d_{max}$ increases, the accuracy of these algorithms decreases. This is because having a large maximum degree indicates that some nodes have a high chance to be connected to the other communities. 
For instance, if a query node is connected to the nodes in its own community, finding a community is relatively easy. However, if the query node is connected to all the nodes in a network, finding a community is relatively hard.
In the experiment, we observe that Fscore returns more optimistic results than NMI and ARI as we have discussed before. Hence, in the following experiments, we do not compare Fscore to avoid misinterpretation. 
\textcolor{black}{
In synthetic networks, we observe that the median NMI score of {\FPA} is up to 6 times higher than the median NMI score of {\huang}, which is more accurate than other baseline algorithms. 
}

Figure~\ref{DMCS_fig:syn_second} shows the efficiency results. In general, {\FPA}, {\kc}, {\kt}, {\kecc}, {\highcore}, and {\hightruss} are comparable in running time. {\NCA} is the slowest because it is required to compute all non-articulation nodes in every iteration. Note that {\FPA} and {\NCA} are designed to solve an NP-hard problem. However, {\kc} and {\kt} are designed for different community search problems, which return results with much lower quality.

\begin{figure}[ht]
\centering
\includegraphics[width=0.99\linewidth]{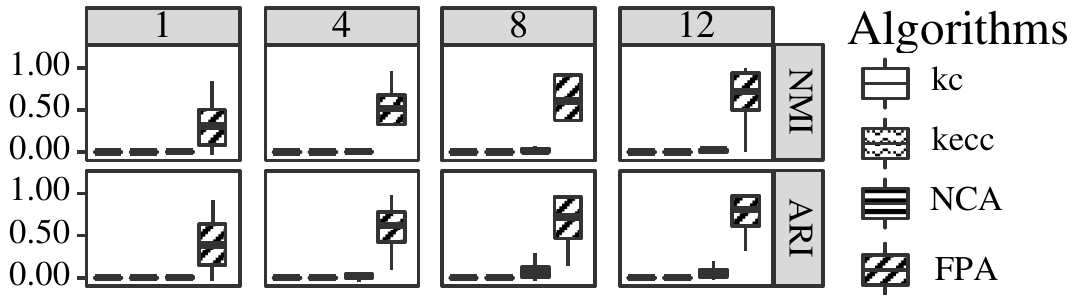}
\vspace{-0.3cm}
\caption{Effect on $|Q|$} 
\label{DMCS_fig:qsize}
\vspace{-0.1cm}
\end{figure}

\subsubsection{Handling multiple query nodes.} 

In Figure~\ref{DMCS_fig:qsize}, we use synthetic networks with default parameters to check effectiveness of {\kc}, {\kecc}, {\NCA}, {\FPA} with various sizes of the query sets. 
Note that we do not include the result of  {\kt}~\cite{huang2014querying} since it allows only a single query node. 
We randomly pick $15$ communities from the ground truth and then, by varying the number of query nodes, randomly pick a set of nodes from each community. 
\textcolor{black}{
We notice that when the size of query set increases, the accuracy of {\kt}, {\NCA}, and {\FPA} increases since query nodes are usually important clues to identify a community. 
This trend is not observed in {\kc} and {\kecc} since both consistently return large communities. 
}

\subsubsection{Scalability test.}
\begin{figure}[ht]
\centering
\includegraphics[width=0.9\linewidth]{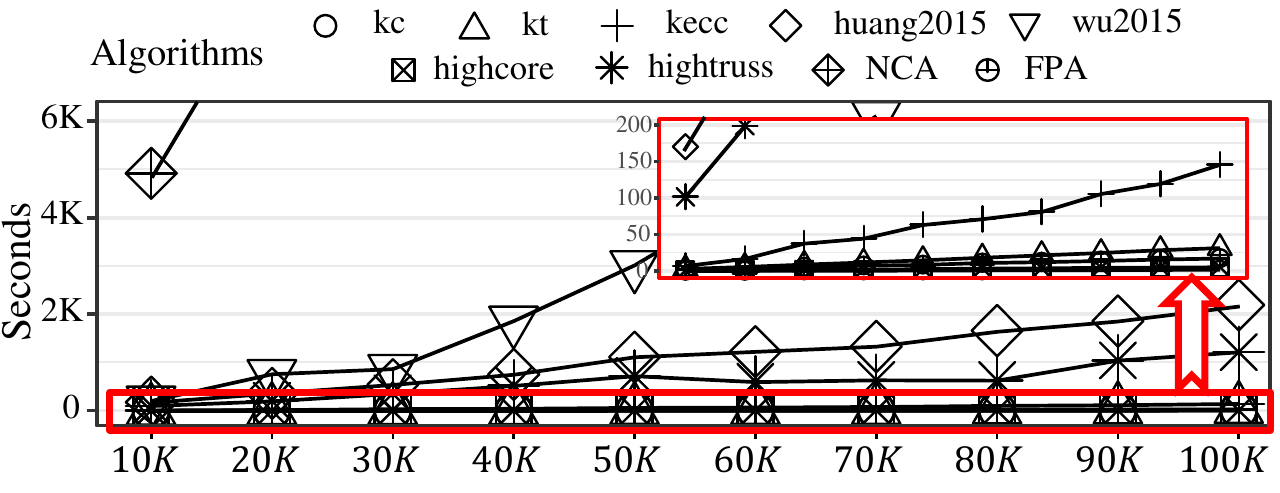}
\vspace{-0.3cm}
\caption{Scalability test}
\vspace{-0.2cm}
\label{DMCS_fig:scalability}
\end{figure}
Figure~\ref{DMCS_fig:scalability} shows the scalability test of our algorithms using synthetic networks by changing the node size from $10K$ to $100K$. 
We observe that {\NCA} is the slowest since {\NCA} is required to compute non-articulation nodes, leading to low efficiency for every iteration.  
We observe that {\FPA} is slower than {\kc} due to some additional operations like sorting and computing the shortest paths but they have the similar scaling trend. 
We observe that {\kc} and {\highcore} scale better than {\FPA} since time complexity of {\kc} (and {\highcore}) and {\FPA} is $O(|V|+|E|)$ and $O(|E|\log{|V|})$, respectively. Note that these algorithms are designed for different community search problems.

\subsubsection{Comparing with other modularity measures.}\label{DMCS:check_diff_mod}
\begin{figure}[h]
\centering
\includegraphics[width=0.99\linewidth]{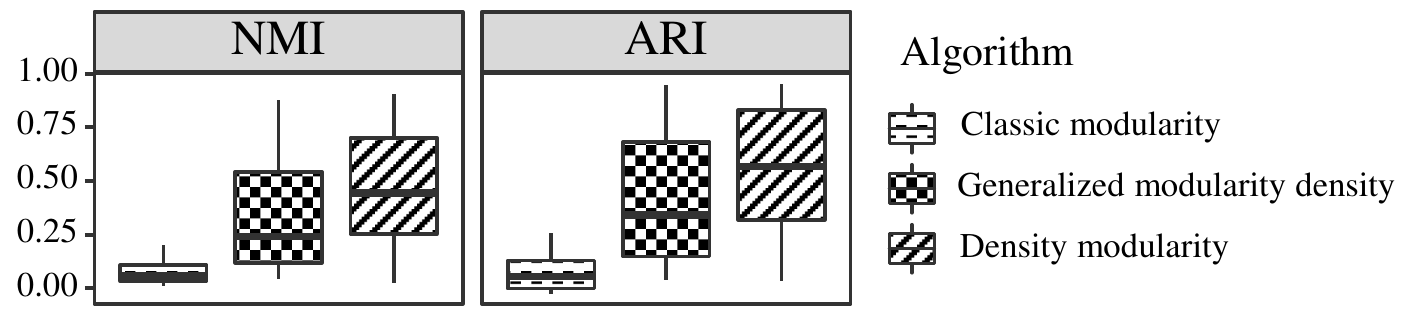}
\vspace{-0.4cm}
\caption{Different modularity scores}
\vspace{-0.2cm}
\label{DMCS_fig:diff_mod}
\end{figure}
We compare the proposed density modularity with the existing modularity measures in terms of the result quality. 
In our algorithms, among all the intermediate subgraphs, we pick a subgraph that has the largest density modularity. 
In Figure~\ref{DMCS_fig:diff_mod}, in addition to the proposed density modularity, we use two different objective functions, namely classic modularity~\cite{newman2006modularity} and generalized modularity density~\cite{guo2020resolution}, in {\FPA} to select the best subgraph.  
\textcolor{black}{For fair comparison, we use different objective functions in the removing process. }

The experiment is performed on a synthetic network with the default parameters. Figure~\ref{DMCS_fig:diff_mod} shows that the proposed density modularity returns a more accurate community compared with both classic modularity and generalized modularity density. 
The reasons are two-fold. First, the density modularity is particularly designed for the community search problem. Second, as we have proved in Section~\ref{DMCS_sec:prob_def}, the density modularity is more prone to escape from free-rider effects. 
\textcolor{black}{As an evidence, our experimental results show that the average size of the communities returned by {\FPA} incorporated with the classic modularity is 18 times larger than that of {\FPA} incorporated with the density modularity.
}

\begin{figure}[ht]
\centering
\includegraphics[width=0.99\linewidth]{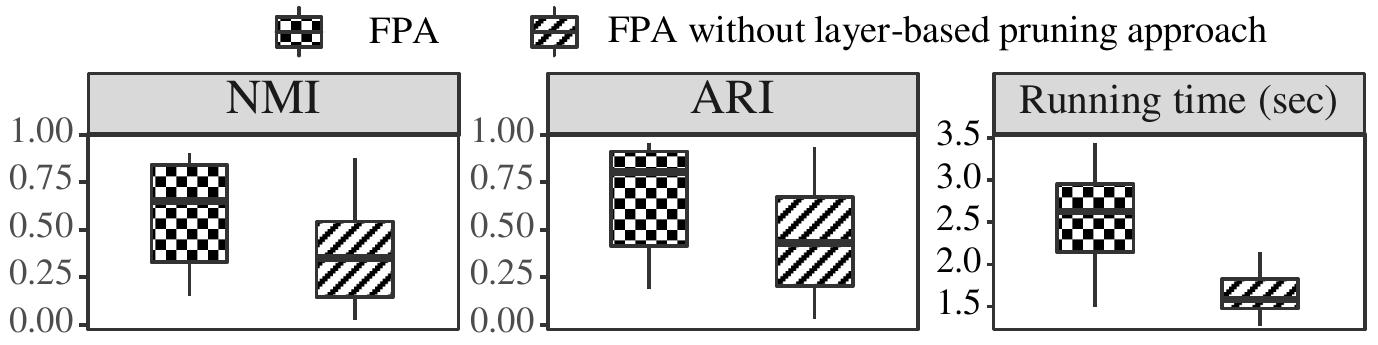}
\vspace{-0.3cm}
\caption{Effect of pruning strategy}
\vspace{-0.3cm}
\label{DMCS_fig:layer_result}
\end{figure}

\subsubsection{Efficiency of layer-based pruning strategy}\label{DMCS_sec:layer_pruning_exp}
\textcolor{black}{
Figure~\ref{DMCS_fig:layer_result} shows effectiveness and efficiency of {\FPA} and {\FPA} without pruning strategy  in a synthetic network with default settings. We observe that effectiveness of {\FPA}  is worse than {\FPA} without pruning strategy since it firstly finds a subgraph having the largest density modularity by iteratively removing distance-based layers then applies removing process. 
We observe utilizing the pruning strategy significantly improves efficiency since it does not consider many nodes which are far from the query nodes.  
}

\begin{figure}[ht]
\centering
\includegraphics[width=0.99\linewidth]{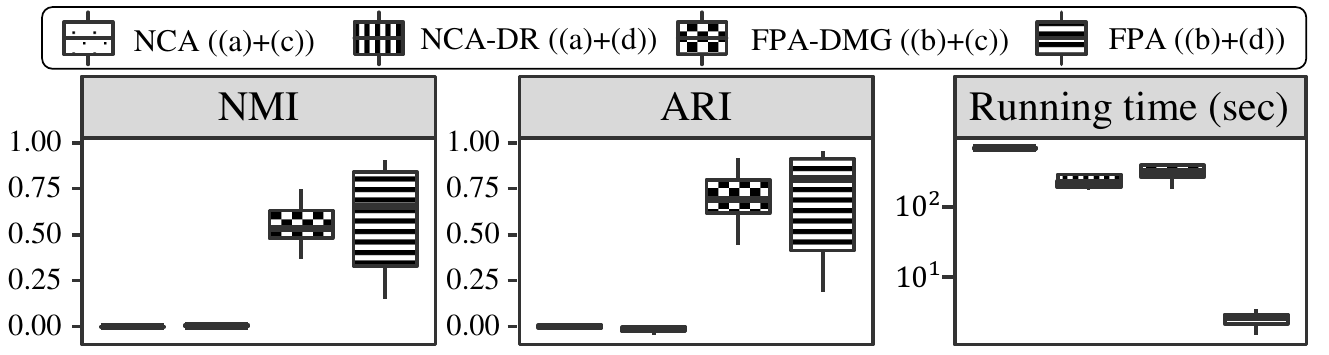}
\vspace{-0.3cm}
\caption{Variations of algorithms}
\vspace{-0.1cm}
\label{DMCS_fig:diff_alg}
\end{figure}

\subsubsection{Variation of algorithms.}\label{DMCS_sec:alg_variation}
In Section~\ref{DMCS_sec:alg}, we discuss two key functions : (1) finding removable nodes; (2) maximizing the density modularity. 
We introduce two approaches for each function, based on which we have {\FPA} and {\FPA}. This experiment is to evaluate the performance of other combinations of the proposed approached: (1) {\NCA}((a)+(c)); (2) {\NCA} with density ratio ({\NCADR}, (a)+(d)); (3) {\FPA} with density modularity gain ({\FPADMG}, (b)+(c)); and (4) {\FPA} ((b)+(d)). Figure~\ref{DMCS_fig:diff_alg} reports the result of effectiveness and efficiency test using a synthetic network with default settings. We observe that {\NCADR} has better efficiency than {\NCA} since computing density ratio is faster than computing density modularity gain. We also observe that {\FPADMG} has comparable accuracy with {\FPA}, but is 150 times slower than {\FPA} since the density modularity gain is unstable. To sum up, {\FPA} is the best among the four algorithms in terms of both efficiency and effectiveness.

\subsection{Performance on Real-World Graphs}\label{DMCS_sec:real_world_experiments}

\begin{figure}[ht]
\centering
\includegraphics[width=0.99\linewidth]{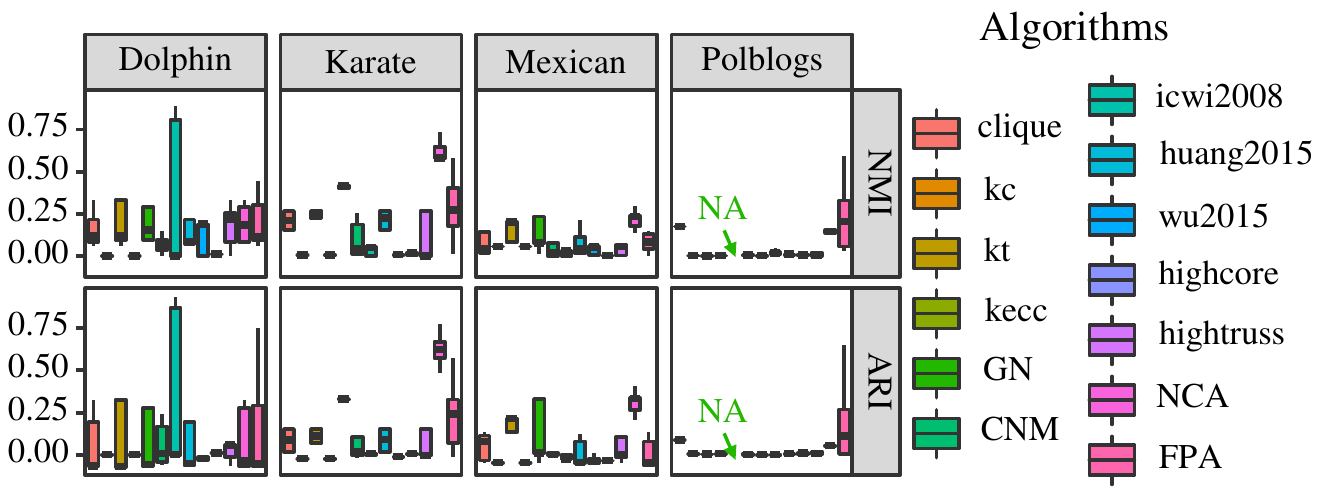}
\vspace{-0.3cm}
\caption{Effectiveness on graphs with distinct communities}
\vspace{-0.4cm}
\label{DMCS_fig:smallreal}
\end{figure}

\begin{figure}[ht]
\centering
\includegraphics[width=0.85\linewidth]{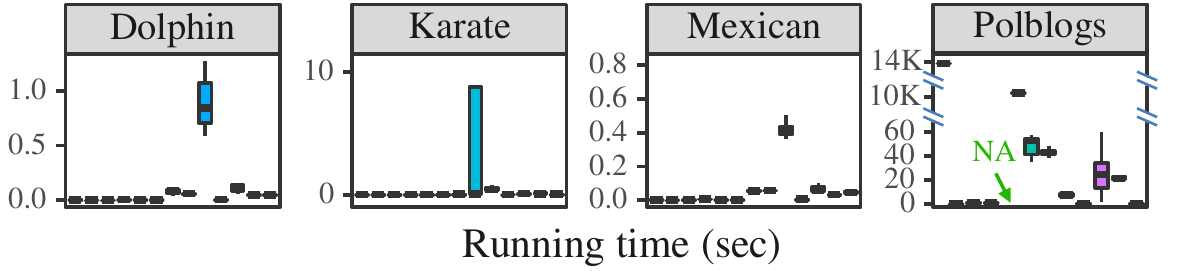}
\vspace{-0.3cm}
\caption{Efficiency 
for Figure~\ref{DMCS_fig:smallreal}}
\vspace{-0.1cm}
\label{DMCS_fig:small_time}
\end{figure}
\spara{Real-world networks with distinct ground-truth communities.}
\textcolor{black}{
Figure~\ref{DMCS_fig:smallreal} shows the results of baseline algorithms and our algorithms in Dolphin, Karate, Mexican, and Polblogs datasets. Each row indicates different measures and each column indicates different datasets.
{\NCA} and {\FPA} outperform other three baseline algorithms significantly in most cases. 
{\NCA} performs very well in Karate and Mexican networks, but less impressive in Dolphin and Polblogs networks. 
This could be because the average difference of the local clustering coefficients of two ground-truth communities is around 10\% in Karate and Mexican while the difference is around 20-50\% in Dolphin and Polblogs. 
The difference implies that there is an unbalance of the community structures in Dolphin and Polblogs networks. 
As discussed in Section~5, 
unbalanced average clustering coefficient can cause problems for {\NCA} to find high-quality results. We also notice that {\FPA} returns the best result in Polblogs network. In Dolphin network, {\hightruss} returns the best result and the results of {\NCA} and {\FPA} are comparable. 
Note that we select the query nodes which belong to $k$-core and $k$-truss. 
i.e., query nodes have the high coreness/trussness value. 
Thus, {\hightruss} and {\highcore} have benefits to get a better-quality result instead of randomly selecting a set of query nodes. {\icwi} returns very unstable results, and mostly it returns very large community as a result because its objective function prefer large-sized community.
Figure~\ref{DMCS_fig:small_time} shows the running time of the baselines and our proposed algorithms. 
{\FPA} is very fast on these datasets although it is slower than {\kc}, {\kt}, and {\kecc} on the three datasets. We also observe that {\clique}, {\wu}, and {\GN} are slower than other baseline algorithms. Note that {\GN} fails to find a solution in Polblogs network within 24 hours.   
}

\begin{figure}[ht]
\centering
\includegraphics[width=0.85\linewidth]{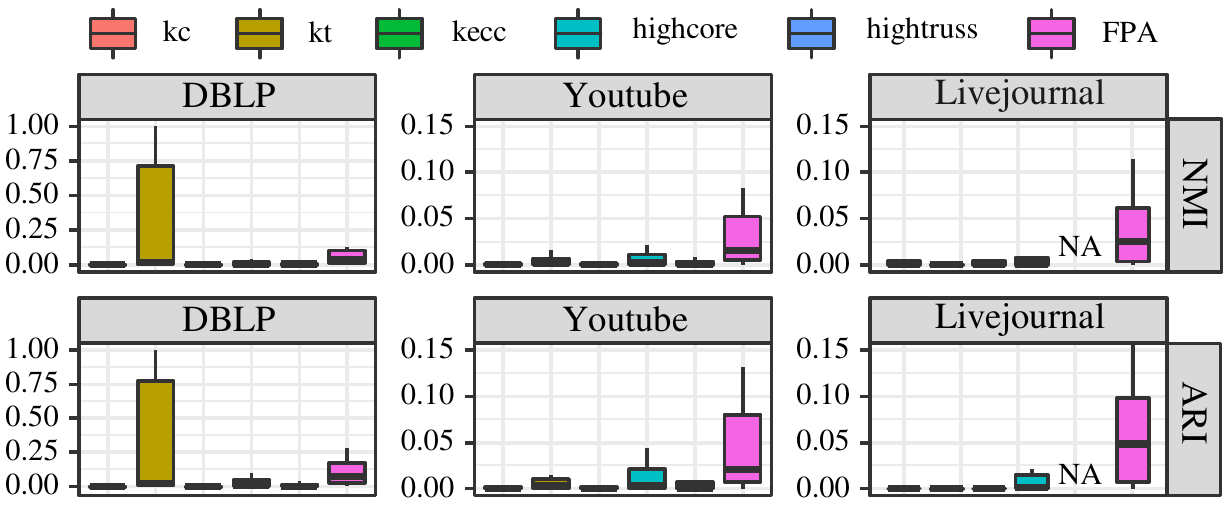}
\vspace{-0.35cm}
\caption{Effectiveness on graphs with overlapping communities.}
\vspace{-0.1cm}
\label{DMCS_fig:real_overlap}
\end{figure}

\spara{Real-world networks with overlapping ground-truth communities.} 
\textcolor{black}{
Figure~\ref{DMCS_fig:real_overlap} shows the evaluation on the real-world networks with overlapping ground-truth community information. Each row indicates different datasets and each column indicates different measures. 
Since there are multiple ground-truth communities that contain the given query node while our result is a single community, we compare our result with each of all the ground-truth communities which contain the query node, and then report the best accuracy. 
We check that real-world networks contain many small-sized communities. Because the {\kc} and {\kecc} return a large-sized community as a result in the datasets, it has relatively lower accuracy than {\FPA}. 
We notice that {\kt} returns fairly small communities in the DBLP dataset since DBLP has less triangles due to its small average degree. In the Youtube dataset, our results show that it returns either a very small or a very large community as a result. In Livejournal dataset, it consistently returns very large communities. 
}

\textcolor{black}{
In DBLP dataset, the median NMI score of {\FPA} is 0.035, which is 2.56 times higher than that of {\kt} having the largest median NMI score among the baseline algorithms. We also observe that the median ARI score of {\FPA} is 0.07, which is 3.16 times higher than the ARI score of {\kt}.
In Youtube dataset,  {\highcore} has higher NMI and ARI than other baseline algorithms. However, we observe that the median NMI score of {\FPA} is 8.5 times higher than the NMI score of {\highcore}, and the median ARI score of {\FPA} is 6.35 times higher than the ARI score of {\highcore}. 
}
\begin{figure}[ht]
\centering
\includegraphics[width=0.99\linewidth]{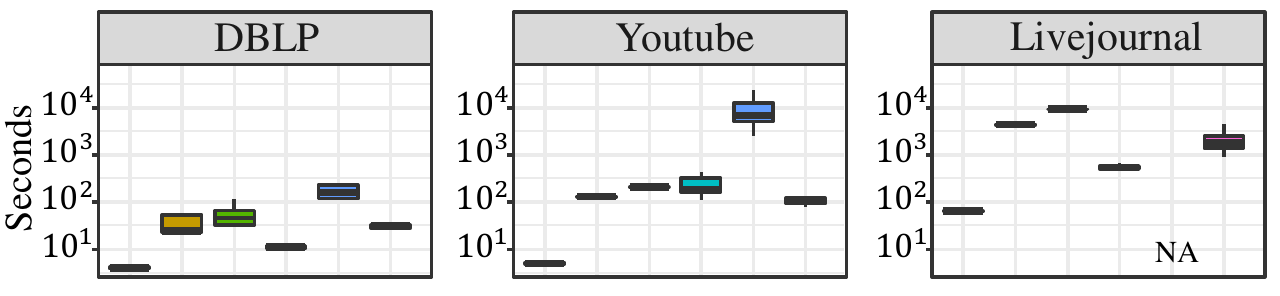}
\vspace{-0.4cm}
\caption{Efficiency 
for Figure~\ref{DMCS_fig:real_overlap}}
\vspace{-0.35cm}
\label{DMCS_fig:real_time}
\end{figure}

We observe that the accuracy values in the datasets are relatively small. This is because (1) the ground-truth communities are overlap; and (2) the average size of the communities is very small. 
The overlapping ground-truth communities make our algorithms find a large-sized community while the size of the communities is relatively small. 
Thus, the accuracy values become relatively small.

\textcolor{black}{
Figure~\ref{DMCS_fig:real_time} shows the running time of baseline algorithms and {\FPA}.
Although {\FPA} is slower than {\kc}, it still can finish in reasonable time on such large datasets. Note again that {\kc} and {\kt} are designed to solve different community search problems, which return much worse results than {\FPA}. 
}

\subsubsection{Varying parameter $k$.} \label{DMCS_sec:changeVarK}
\begin{figure}[t]
\centering
\includegraphics[width=0.99\linewidth]{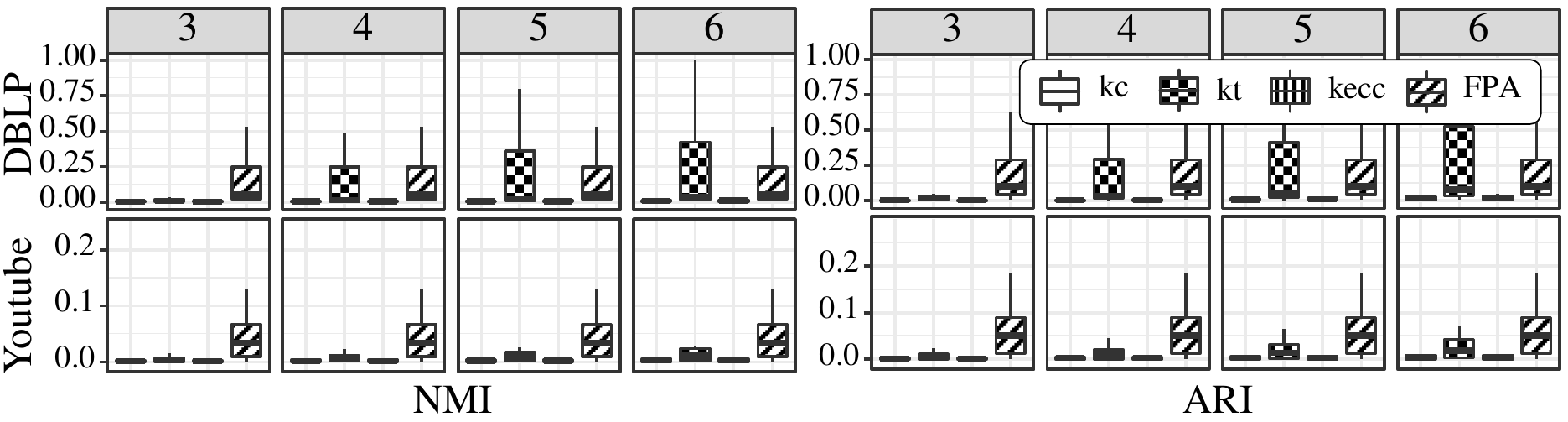}
\vspace{-0.35cm} 
\caption{Effect on $k$} 
\label{DMCS_fig:ksize}
\vspace{-0.1cm}
\end{figure}
\textcolor{black}{
In Figure~\ref{DMCS_fig:ksize}, we use DBLP and Youtube datasets with default parameters to check the result for various user parameter $k$, which is shown in the top of the figure. Each row corresponds to the same dataset and each column corresponds to a specific $k$ value. 
{\kc} and {\kecc} fail to achieve high accuracy in these cases. 
For {\kt}, it returns the best result when $k=5$ or $k=6$. For all the cases, {\FPA} outperforms baseline algorithms. It implies that  classic community search models such as {\kc}, {\kecc}, and {\kt} need a proper parameter to obtain a high-quality result. In contrast, our community model does not require any user parameters. }

\subsubsection{Case study.  }
We conduct a case study to show the usefulness of {\DMCS}. 
In DBLP dataset~\cite{kim2014link}, a node indicates an author, and an edge indicates that the two authors have co-authored at least 3 times. 
We set the query node $Q=\{\text{Philip S. Yu}\}$ and check the result of our {\FPA}, $3$-truss, and $3$-core.
We notice that our algorithm returns a small-sized community and all the authors in the community are connected to the query node. However, $3$-truss contains $157$ authors, and the query node is connected to $17\%$ of the authors in the community. $3$-core contains $1,040$ authors and the query node is connected to only $1\%$ of the authors in the community.

When we compute Betweenness centrality~\cite{brandes2001faster} and Eigen centrality~\cite{zaki2014data}, the query node has the largest centrality scores in our community.
However, in $3$-truss, we observe that the query node has the second largest centrality values among the nodes, and in $k$-core, the query node is ranked 45th in Betweenness centrality and 175th in Eigen centrality.

\begin{figure}[ht]
\centering
\subfloat[Our result]{\includegraphics [width=0.32\columnwidth]{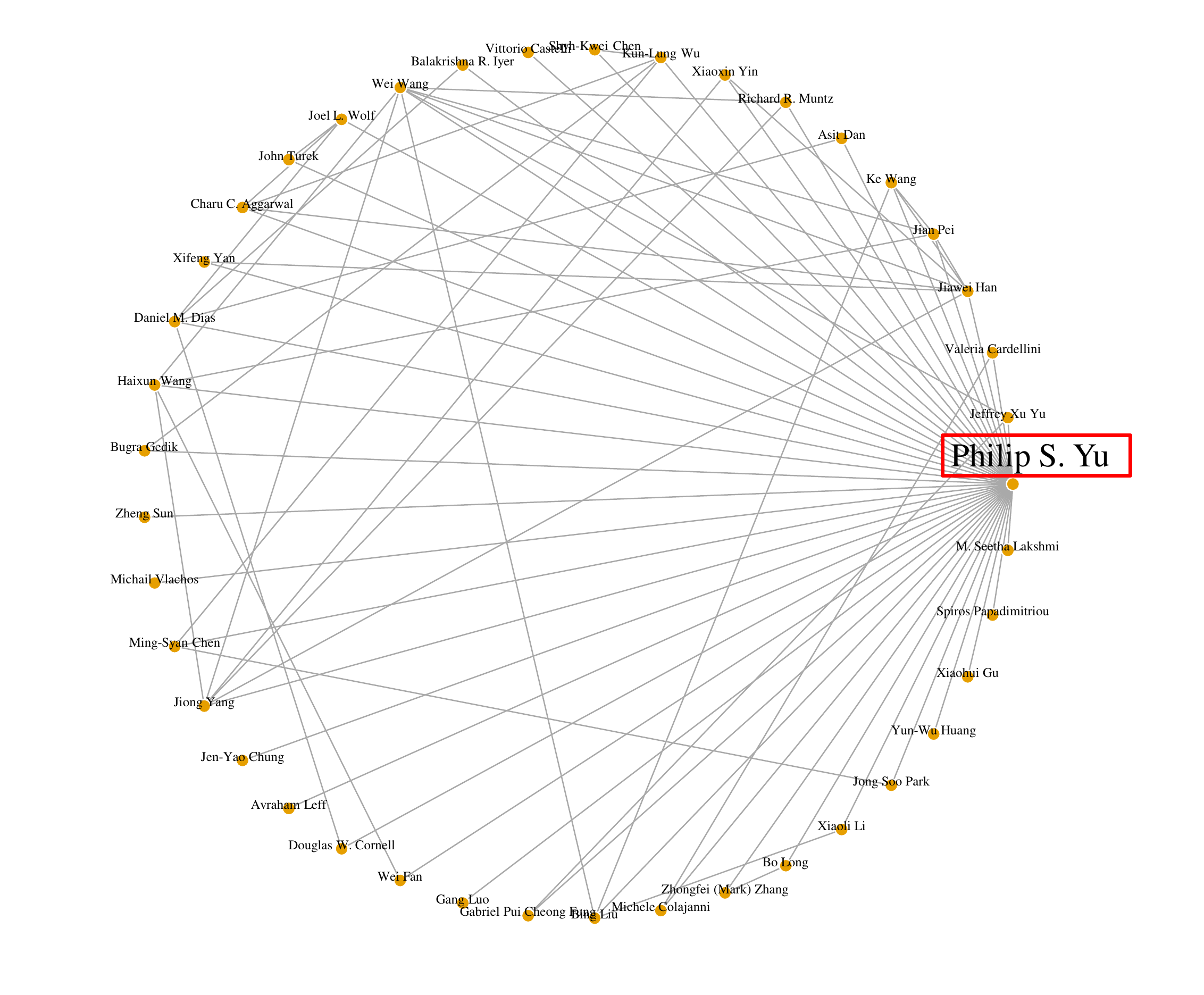}
\label{DMCS_fig:our-dblp}}
\hfil
\subfloat[$3$-truss]{\includegraphics [width=0.32\columnwidth]{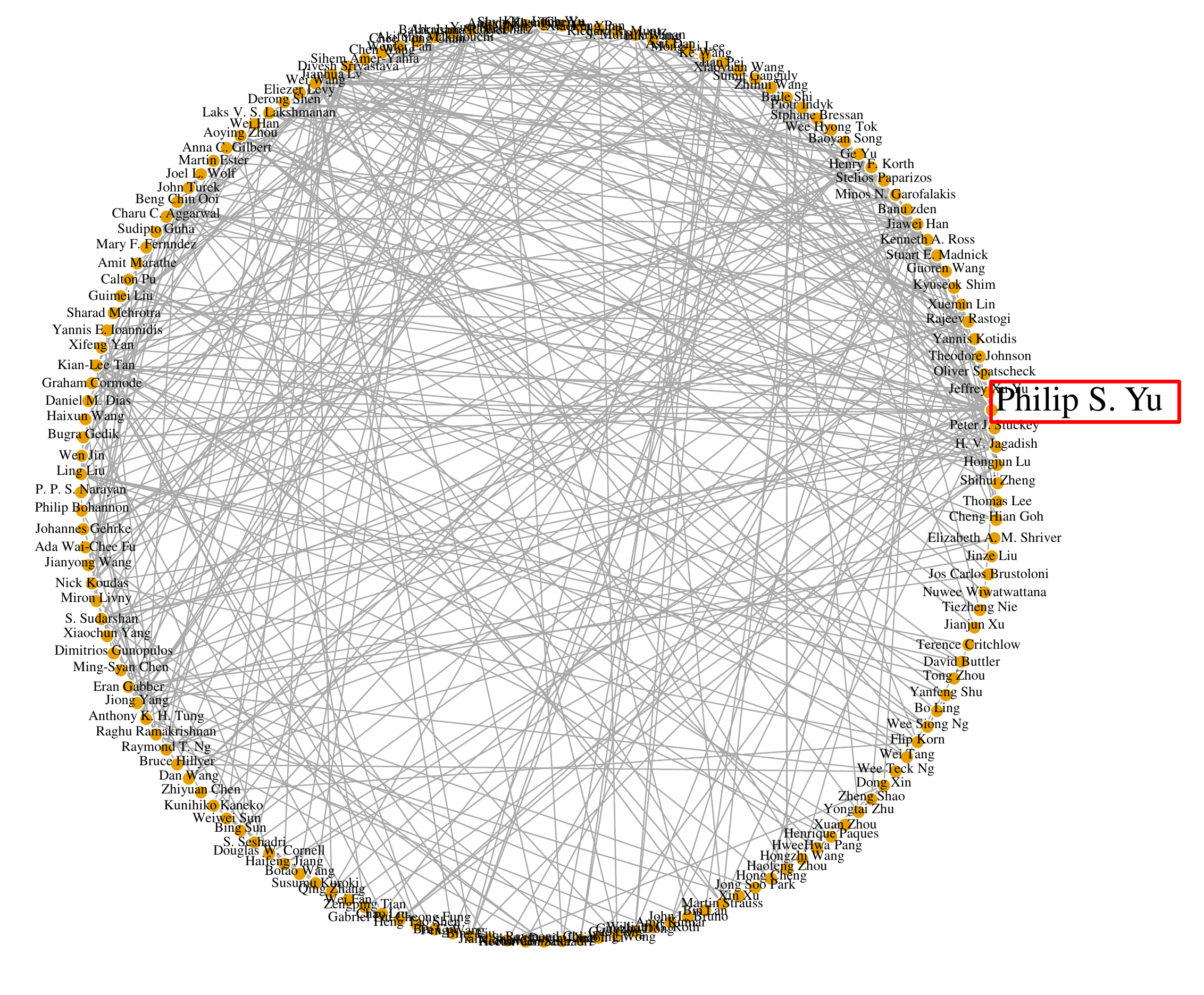}
\label{DMCS_fig:truss-dblp}}
\hfil
\subfloat[$3$-core]{\includegraphics [width=0.32\columnwidth]{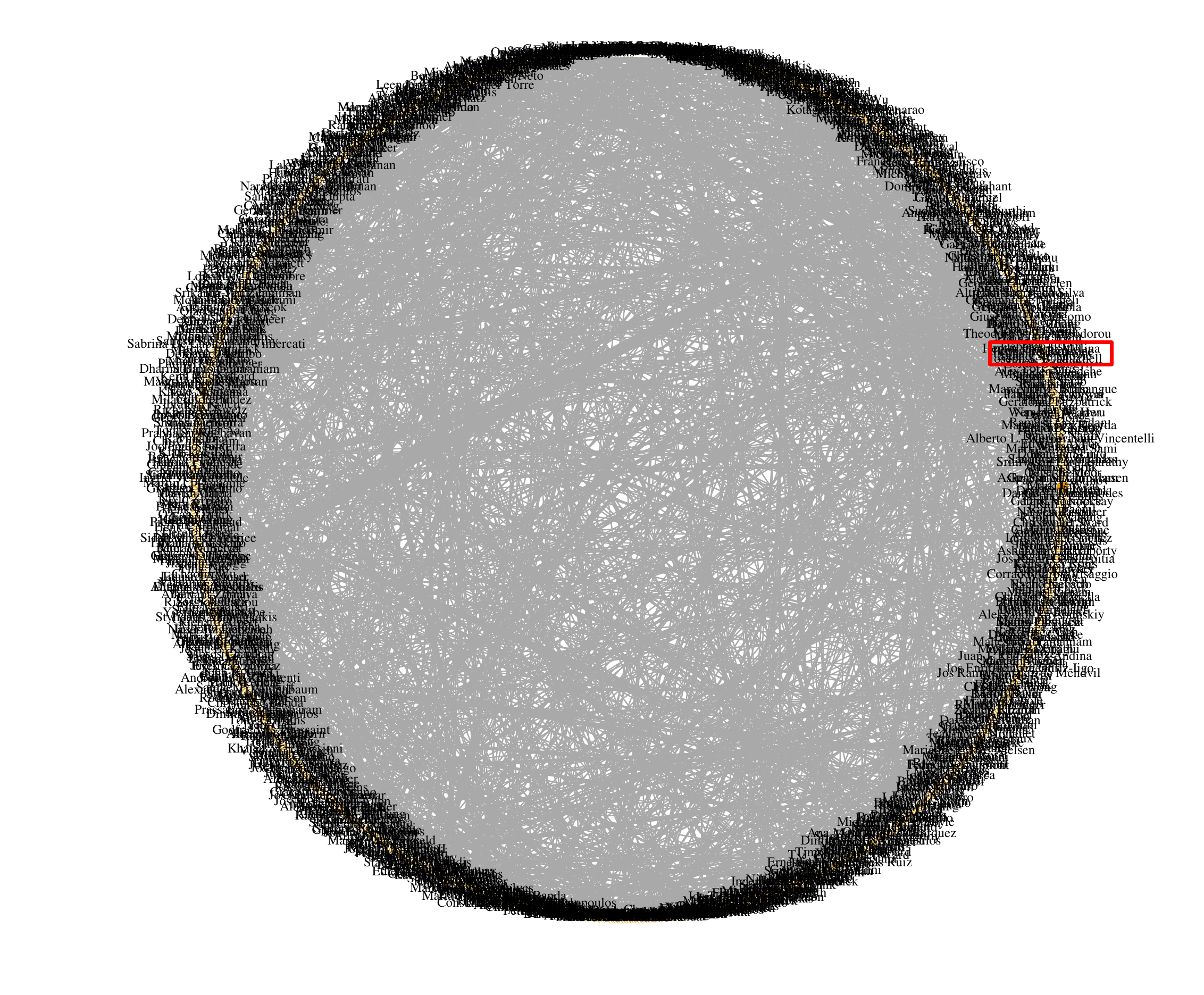}
\label{DMCS_fig:core-dblp}}
\vspace{-0.3cm}
\caption{Case study}
\vspace{-0.4cm}
\label{DMCS_fig:case}
\end{figure}

\section{Conclusion}\label{DMCS_sec:conclus}

In this work, we propose modularity-based community search (\DMCS) problem that aims to find a community which is densely connected internally and sparsely connected externally in a network while containing all the query nodes. We define a new modularity score called density modularity and prove the superiority of the density modularity compared with the classic modularity for the community search problem. We prove {\DMCS} problem is NP-hard, and design two algorithms for the problem. 
We conduct extensive experimental studies on large-scale real-world networks to demonstrate the efficiency and effectiveness of our proposed algorithms. 
In the future, we can utilize our new density modularity to solve the community detection problem since the density modularity can mitigate the resolution limit problem.  

\appendix

\section{Proof for Lemma~\ref{DMCS_lemma:DM_rlp_CM}}\label{DMCS_appendix:FRE_lemma1}
 
Let $S$ be a community and $S^*$ be the optimal community. Note that all the classic modularity and density modularity values for identified communities are positive. Otherwise, the identified community is meaningless. 
Let $l_{int}$ be the set of intersected internal edges between $S$ and $S^*$; let $d_{int}$ be the sum of degree in the intersected nodes between $S$ and $S^*$; let $S_{int}$ be the common nodes of $S$ and $S^*$. We have equations of $CM(S)$, $CM(S\cup S^*)$, $DM(S)$, and $DM(S\cup S^*)$ are defined as follows. 

\begin{itemize}[leftmargin=*]
    \item $CM(S)$ : $\frac{l_S}{|E|} - \frac{d_S^2}{4|E|^2}$ $\phantom{123123123} \bullet$ $DM(S)$ : $\frac{l_S}{|S|} - \frac{d_S^2}{4|E||S|}$
\end{itemize}
\begin{itemize}[leftmargin=*]
    \item $CM(S \cup S^*)$ : $\frac{l_S + l_{S^*} -l_{int}}{|E|} - \frac{(d_S + d_{S^*} - d_{int})^2}{4|E|^2}$
    \item $DM(S \cup S^*)$ : $\frac{l_S + l_{S^*}-l_{int}}{|S|+|S^*|-|S_{int}|} - \frac{(d_S + d_{S^*} - d_{int})^2}{4|E|(|S|+|S^*|-|S_{int}|)}$
\end{itemize}

To avoid the free-rider effect for the classic modularity, we have

\begin{align}
&CM(S) \geq CM(S\cup S^*) \notag \\
\Leftrightarrow&\frac{l_S}{|E|} - \frac{d_S^2}{4|E|^2} \geq \frac{l_S + l_{S^*} - l_{int}}{|E|} - \frac{(d_S + d_{S^*} - d_{int})^2}{4|E|^2} \nonumber\\
\Leftrightarrow 
& 0 \geq   X\label{DMCS_eq:cm_bound}
\end{align}
where $X=4|E|(l_{S^*} -l_{int}) -(2d_Sd_{S^*} - 2d_Sd_{int} + d_{S^*}^2 - 2d_{S^*}d_{int} + d_{int}^2)$.

To avoid the free-rider effect for the density modularity, we have
\begin{align}
    &DM(S) \geq DM(S\cup S^*) \notag \\
\Rightarrow&\frac{l_S}{|S|} - \frac{d_S^2}{4|E||S|} \geq \frac{l_S + l_{S^*} - l_{int}}{|S|+|S^*|-|S_{int}|} -  \frac{(d_S + d_{S^*}-d_{int})^2}{4|E|(|S|+|S^*|-|S_{int}|)}\nonumber\\
   \Rightarrow  &\begin{aligned}
        0 \geq & X   - \frac{4l_S |E|(|S^*|-|S_{int}|) + d_S^2 (|S_{int}|-|S^*|)}{|S|}
    \end{aligned} \label{DMCS_eq:mod_bound}
\end{align}

\textcolor{black}{
We notice that the difference between Equations~\ref{DMCS_eq:cm_bound} and \ref{DMCS_eq:mod_bound} is the term $T=\frac{4l_S |E|(|S^*|-|S_{int}|) + d_S^2 (|S_{int}|-|S^*|)}{|S|}$ which is positive. 
\begin{align}\label{DMCS_eq:supplementary_term_test1}
    &\frac{4l_S |E|(|S^*|-|S_{int}|) + d_S^2 (|S_{int}|-|S^*|)}{|S|}  
    =\frac{(4l_S |E| - d_S^2)(|S^*|-|S_{int}|)}{|S|}\nonumber\\
    &=\frac{4|E|^2(|S^*|-|S_{int}|)}{|S|} {\left(\frac{l_S}{|E|}-\frac{d_S^2}{4|E|^2}\right )}=\frac{4|E|^2(|S^*|-|S_{int}|)}{|S|} CM(S)>0\nonumber
\end{align}
The last inequality holds because $CM(S)>0$ and $|S^*|-|S_{int}|>0$.
Thus, when the density modularity suffers from the free-rider effect, the classic modularity suffers from the free-rider effect as well. It directly implies that escaping the free-rider effect for the classic modularity is harder than density modularity since the Equation~\ref{DMCS_eq:mod_bound} is easier to satisfy compared with Equation~\ref{DMCS_eq:cm_bound}.
}

\section{Proof for Lemma~\ref{DMCS_lemma:RL}}\label{DMCS_app:RL}
We show that {when density modularity} suffers from the resolution limit problem, the classic modularity suffers from the resolution limit problem as well. 
 
To illustrate the resolution limit problem, we reuse the observation in Lemma~\ref{DMCS_lemma:DM_rlp_CM}. In Appendix~\ref{DMCS_appendix:FRE_lemma1}, Equations~\ref{DMCS_eq:cm_bound} and \ref{DMCS_eq:mod_bound} are generalized forms to avoid the resolution limit problem for the community search. Note that $|S_{int}|$ and $d_{int}$ are $0$ since $S_1$ and $S_2$ do not overlap based on the definition of the resolution limit problem (See Definition~\ref{DMCS_def:RLP}). To overcome resolution limit problem for both CM and DM, we get the following inequality. 

\begin{align}
\begin{aligned}
&CM(S) \geq CM(S\cup S^*) \\
\Rightarrow& (l_S - \frac{d_S^2}{4|E|}) \geq (l_S + l_{S^*} - l_{int} - \frac{(d_S + d_{S^*})^2}{4|E|}) \phantom{1231231231}\\
\end{aligned}
\end{align}

\begin{align}
\begin{aligned}
    &DM(S) \geq DM(S\cup S^*)  \\
\Rightarrow& \frac{1}{|S|}(l_S - \frac{d_S^2}{4|E|}) \geq \frac{1}{|S|+|S^*|}(l_S + l_{S^*} - l_{int} -  \frac{(d_S + d_{S^*})^2}{4|E|})\\
\end{aligned}
\end{align}

\textcolor{black}{
In the above Equations, $\frac{1}{|S|}$ is much larger than $\frac{1}{|S|+|S^*|}$. Thus, it indicates that when the classic modularity does not suffer from the resolution limit problem, the density modularity does not suffer the resolution limit problem too. This also implies that when the density modularity suffers from the resolution limit problem, the classic modularity also suffers from the resolution limit problem. 
}

\section{Proof for Theorem~\ref{DMCS_theorem:NP}}\label{DMCS_app:NP}

In this proof, we reduce an instance of the set-cover problem to an instance of DMCS problem. 

Suppose that we have an instance $I_{SC}=\{I,S\}$  of the set-cover problem, where $I$ is a set of items, $S$ is sets of items whose union is equal to $I$. 
We first create four graphs. 
\begin{enumerate}[leftmargin=*]
    \item $B_1$ : A bipartite graph $B_1$ is constructed by $I_{SC}$. It contains two node sets $U=I$ and $V=S$, and edges $B_E$ connecting a node $u\in U$ to $v\in V$ if an item $u$ belongs to a set $v$. 
    Without loss of generality, we make $|U|=|V|$ by adding dummy items or sets. 
    For example, if $|U|<|V|$, we can add $|V|-|U|$ nodes to $U$ and then make connections from $|V|-|U|$ nodes to all nodes in $V$, i.e., we add $|V|(|V|-|U|)$ edges. If $|U|>|V|$, we add $|U|-|V|$ dummy nodes to $V$ without any connections. 
    \item $B_2$ : We construct a bipartite network $B_2$ consisting of two disjoint sets $V$ in $B_1$ and $T$ which contains $|V|^2$ nodes. 
    Every node $v\in V$ has exact $|V|$ neighbor nodes in $T$ and every node $t\in T$ has exact one neighbor node.  
    \item $G_1$ : We next make a graph $G_1$ which contains $U$ nodes in $B_1$ and in which every node has own self edge. 
    \item $B_3$ : We construct a bipartite graph $B_4$ consisting of a query node $q$ and $V$ in $B_1$, and $q$ is connected to all nodes in $V$. 
\end{enumerate}
Now, we construct an instance $I_{DMCS}=(G=\{B_1 \cup B_2\cup G_1\cup B_3\}, Q=\{q\cup U\})$. 
Figure~\ref{DMCS_fig:NP} shows the constructed graph.

\begin{figure}[t]
\centering
\includegraphics[width=0.99\linewidth]{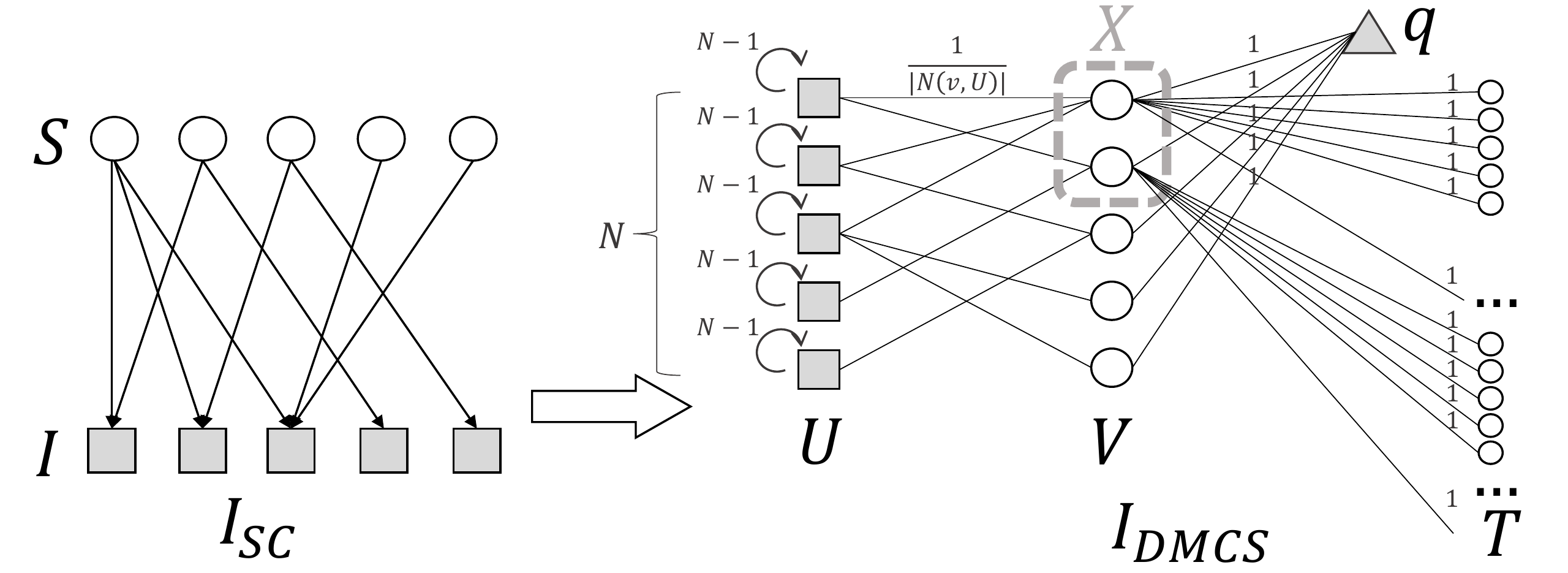}
\vspace{-0.1cm}
\caption{A reduction from Set-cover problem to {\DMCS}}
\label{DMCS_fig:NP}
\end{figure}

We are ready to start the reduction process. Note that our DMCS problem aims for finding a community which satisfies the connectivity constraint while containing all the query nodes.  
Suppose that our community $C = Q$. Currently, $C$ is not connected. Thus, we need to choose the nodes in the $V$ side to guarantee the connectivity. 
It means that our resultant community $C$ has to include $Q$ and some nodes in $V$. We denote a set of nodes $X \subseteq V$ if it belongs to the community $C$, and denote $N$ the number of nodes in $U$. 

There are two possible scenarios. 
The first scenario is that our resultant community $C$ does not include any nodes in $T$, i.e., our resultant community $C$ includes the nodes $U$, $q$, and $X$. 
The other scenario is that our resultant community $C$ consists of $U$, $q$, $X$, and some nodes in $T$ which are connected to $X$. 

\spara{Scenario 1.} 
The values of the variables of the density modularity for Scenario 1 are as follows. 
\begin{itemize}[leftmargin=*]
    \item $|C|= 1+N+|X|$
    \phantom{1123123.} $\bullet$ $w_C = 2|X| + N(N-1)$
    \item $w_G = 2N + N^2 + N(N-1)$
    \phantom{12} $\bullet$ \mbox{$d_C = N + 2|X| + N|X| + N + N(N-1)$}
\end{itemize} 

We check the derivative of $DM_1(G, |X|)$ to check whether $DM_1$ is monotonic decreasing or not. 

\begin{align}
    \frac{\partial DM_1(G, |X|)}{\partial |X|}=& -\frac{(N^2+4N+4)|X|^2 + (2N^3+10N^2+16N+8)|X|}{4N(2N+1)(|X|+N+1)^2} \notag  \\ 
    & - \frac{9N^4 -14N^3 - 19N^2 -4N}{4N(2N+1)(|X|+N+1)^2}
\end{align}
\spara{Scenario 2.} 
We check the case when the resultant community $C$ contains the nodes in $T$. The values of the variables of the density modularity for Scenario 2 are as follows.
\begin{itemize}[leftmargin=*]
    \item $|C|=|q\cup U \cup X| = 1+N+|X|+N|X|$
    \item $w_C = 2|X| + N(N-1) + N|X|$
    \item $w_G = 2N + N^2 + N(N-1)$
    \item $d_C = N + 2|X| + N|X| + N + N(N-1) + N|X|$
\end{itemize} 

Similarly, we check the derivative. 
\begin{align}
\frac{\partial DM_2(G, |X|)}{\partial |X|} =& -\frac{(4N^2+8N+4)|X|^2 + (8N^2 + 16N + 8)|X|}{4N(N+1)(2N+1)(|X|+1)^2} \notag \\
&-\frac{(7N^4 - 10N^3 - 17N^2 - 4N)}{4N(N+1)(2N+1)(|X|+1)^2}
\end{align}

We observe that $\frac{\partial DM_1(G, |X|)}{\partial |X|} < 0$ and $\frac{\partial DM_2(G, |X|)}{\partial |X|} < 0$ since the dominant term is $N$ and $N$ can be easily expanded by adding dummy nodes. Thus, the objective functions $DM_1$ and $DM_2$ always decrease. 
Note that the size of $X$ must be larger than or equal to $1$ since all the query nodes must be connected. 
If we find an optimal solution of DMCS, it means that we can find a minimal set $X$ which can maximize the density modularity. 
We notice that $X$ can be an optimal solution of the set-cover problem since it is connected to all the nodes in the $U$ side, which are the items in $I_{SC}$, and its size is minimized.
Therefore,
finding a solution of an instance $I_{DMCS}$ is the same as finding a solution of $I_{SC}$. Therefore, we have the proof.

\bibliographystyle{ACM-Reference-Format}
\bibliography{acmart}

\end{document}